\documentclass[a4paper,11pt]{article}
\usepackage{booktabs}

\usepackage[T1]{fontenc}
\usepackage[sc]{mathpazo}
\usepackage{amsmath,amsfonts,amssymb,amsthm}
\usepackage[mathscr]{eucal}
\usepackage[square,numbers]{natbib}
\usepackage{tgpagella}
\usepackage[square]{natbib}
\theoremstyle{plain}
\newtheorem{theorem}{Theorem}
\newtheorem{lemma}[theorem]{Lemma}
\newtheorem{proposition}[theorem]{Proposition}

\newtheoremstyle{note}{\topsep}{\topsep}{\slshape}{}{\scshape}{}{ }{}
\theoremstyle{note}

\newtheorem{remark}[theorem]{Remark}
\newtheorem{example}[theorem]{Example}
\numberwithin{equation}{section}
\numberwithin{theorem}{section}
\newcommand\cX{{\mathcal X}}

\newcommand\scD{{\mathscr D}}

\newcommand\scI{{\mathscr I}}

\newcommand\scM{{\mathscr M}}
\newcommand\scN{{\mathscr N}}

\newcommand\scP{{\mathscr P}}

\newcommand\scR{{\mathscr R}}

\newcommand\scV{{\mathscr V}}

\newcommand\scX{{\mathscr X}}
\newcommand\scY{{\mathscr Y}}

\newcommand\mvector{\boldsymbol}

\newcommand\va{\mvector{a}}

\newcommand\vd{\mvector{d}}

\newcommand\vp{\mvector{p}}
\newcommand\vq{\mvector{q}}

\newcommand\vv{\mvector{v}}

\newcommand\vx{\mvector{x}}

\newcommand\vA{\mvector{A}}

\newcommand\vC{\mvector{C}}
\newcommand\vD{\mvector{D}}
\newcommand\vE{\mvector{E}}

\newcommand\vX{\mvector{X}}

\newcommand\veta{\mvector{\eta}}

\newcommand\vGamma{\mvector{\Gamma}}

\newcommand\vzero{\mvector{0}}

\newcommand\field{\mathbb}

\newcommand\R{\field{R}}

\newcommand\C{\field{C}}
\newcommand\Z{\field{Z}}

\newcommand\N{\field{N}}
\newcommand\Q{\field{Q}}

\newcommand\bbP{\mathbb{P}}

\newcommand\PP{\mathbb{P}}

\newcommand\bLambda{\boldsymbol{\Lambda}}

\newcommand\tr{\operatorname{Tr}}
\newcommand\res{\operatorname{res}}

 \newcommand\mult{\operatorname{mult}}

\newcommand\grad{\operatorname{grad}}

\newcommand\rmd{\mathrm{d}}

\newcommand\CP{\ensuremath{\C\bbP}}

\newcommand\rmi{\mathrm{i}\mspace{1mu}}

\newcommand\Dt{\frac{\mathrm{d}\phantom{t} }{\mathrm{d}\mspace{1mu} t}}

\newcommand\pder[2]{\dfrac{\partial #1 }{\partial #2}}
\newcommand\abs[1]{\lvert #1 \rvert}

\newcommand\mtext[1]{\quad\text{#1}\quad}

\newcommand\defset[2]{\left\{{#1}\;\vert \;\; {#2} \,\right\}}

\title{Darboux Points and Integrability Analysis of Hamiltonian Systems with Homogeneous Rational Potentials}
\author{Micha\l{} Studzi\'nski\\
Institute for Theoretical Physics and Astrophysics,\\
University of Gda{\'n}sk\\
Wita Stwosza 57,
PL-80-952 Gda{\'n}sk, Poland
\\ and \\
National Quantum Information Centre of Gda\'nsk\\
W\l{}. Andersa 27,
PL-81-824 Sopot, Poland\\
e-mail:  studzinski.m.g@gmail.com
\\ and \\
Maria Przybylska \\
Institute of Physics, University of Zielona G\'ora, \\
 Licealna 9, PL-65--417,  Zielona G\'ora, Poland
\\ e-mail: M.Przybylska@proton.if.uz.zgora.pl
}

\begin{document}
%
%

\maketitle

\begin{abstract}
  We study the integrability in the Liouville sense of natural
  Hamiltonian systems with a homogeneous rational potential $V(\vq)$.
  Strong necessary conditions for the integrability of such systems
  were obtained by an analysis of differential Galois group of
  variational equations along certain particular solutions. These
  conditions have the form of arithmetic restrictions putted on
  eigenvalues of Hessian $V''(\vd)$ calculated at a non-zero solution
  $\vd$ of equation $\grad V(\vd)=\vd$. Such solutions are called
  proper Darboux points.

It was recently proved that for generic polynomial homogeneous
  potentials there exist universal relations between eigenvalues of
  Hessians of the potential taken at all proper Darboux points.  The
  question about the existence of such relations for rational
  potentials seems to be hard. One of the reason of this fact is the
  presence of indeterminacy points of the potential and its
  gradient. Nevertheless, for two degrees of freedom we prove that
  such relation exists. This result is important because it allows to
  show that the set of admissible values for eigenvalues of Hessian at  a proper
  Darboux point for potentials satisfying necessary conditions for the
  integrability is finite. In turn, it gives a tool for classification
  of integrable rational potentials.
 
Also, quite recently, it was shoved that for polynomial homogeneous potentials
additional necessary conditions for the integrability can be deduced from the
existence of
improper Darboux points, i.e., points $\vd$ which are non-zero solution of
equation $\grad V(\vd)=\vzero$. These new conditions have also the form of
arithmetic restrictions imposed on eigenvalues of $V''(\vd)$.  In this paper we 
prove that for rational potentials improper Darboux points give the same
necessary conditions for the integrability.


\end{abstract}

\section{Introduction and main results}
\label{sec:intro}
In this paper we consider Hamiltonian systems given by Hamilton's
function of the form
\begin{equation}
  \label{ham}
  H=\frac{1}{2}\sum_{i=1}^np_i^2+V(\vq),
\end{equation} 
where $\vq=(q_1,\ldots,q_n)$ and $\vp=(p_1,\ldots,p_n)$ are canonical
coordinates and momenta in the complex phase space $\C^{2n}$. We assume that 
potential $V=V(\vq)\in\C(\vq)$ is a homogeneous rational function.  We
can write it the following form
\begin{equation}
  V(\vq):=\dfrac{W(\vq)}{U(\vq)},
  \label{eq:ratpot}
\end{equation} 
where $W(\vq)$ and $U(\vq)$ are relatively prime homogeneous
polynomials of degrees $r,s\in\N_{0}:=\N\cup\{0\}$,
respectively. Hence, potential $V(\vq)$ is a homogeneous function of
degree $k:=r-s\in\Z$.

Our aim is to investigate the integrability of such systems.
Recently, this problem was considered in papers \cite{mp:04::d,mp:05::c}
for systems with polynomial homogeneous potentials.  Methods developed
in these papers are strong. Thanks to them it was possible to give
necessary and sufficient conditions for the integrability of
homogeneous polynomial potentials of small degree $k$, and small
number of freedom $n$. The main purpose of this paper is to extend
methods and results of the cited papers to systems with rational
potentials. Some first attempts for such systematic analysis are contained in
B.Sc. thesis and M.Sc. thesis of the first author
\cite{Studzinski:09::,Studzinski:10::}.

The equations generated by Hamiltonian~ \eqref{ham} have the canonical
form
\begin{equation}
  \label{eq:ham}
  \Dt \vq=\vp,\qquad \Dt \vp=-V'(\vq),
\end{equation} 
where $V'(\vq):=\grad V(\vq)$. We assume that the time $t$ is complex.
We say that potential $V$ is integrable if this system is integrable
in the Liouville sense, i.e., equations \eqref{eq:ham} possess $n$
functionally independent commuting first integrals. It is useful to
consider the equivalence classes of potentials as it was made
in~\cite{mp:05::c}. Namely, let $\mathrm{PO}(n,\C)$ be the complex
projective orthogonal subgroup of $\mathrm{GL}(n,\C)$, i.e.,
\begin{equation}
  \mathrm{PO}(n,\C)=\{\vA\in\mathrm{GL}(n,\C),\ |\ \vA\vA^T=\alpha \vE_n,\;
  \alpha\in\C^\star\},
\end{equation}
where $\vE_n$ is $n$-dimensional identity matrix.
We say that $V$ and $\widetilde V$ are equivalent if there exists
$\vA\in \mathrm{PO}(n,\C)$ such that $\widetilde
V(\vq)=V_{\vA}(\vq):=V(\vA\vq)$. Obviously, if $V$ is integrable, then
also $V_{\vA}$ is integrable, for any $A\in \mathrm{PO}(n,\C)$.  In
further considerations a potential means a class of equivalent
potentials in the above sense.

The crucial role in our consideration plays the notion of Darboux
point.  A non-zero $\vd\in\C^n$ is called a Darboux point of potential
$V$ iff $V'(\vd)$ is parallel to $\vd$. Thus, $\vd$ is a solution of
the following system
\begin{equation}
  \label{propDar}
  V'(\vd)=\gamma\vd,
\end{equation} 
where $\gamma\in\C$, or, equivalently, it satisfies
\begin{equation}
  \label{eq:dar2}
  \vd\wedge V'(\vd)=0.
\end{equation}

We say that a Darboux point $\vd$ is a proper Darboux point iff
$V'(\vd)\neq \vzero$, otherwise it is called an improper one.

Proper Darboux point defines a straight-line
particular solution of the system \eqref{eq:ham} given by
\begin{equation}
  (\vp(t),\vq(t))=(\varphi(t)\vd,\dot{\varphi}(t)\vd),
  \label{eq:party}
\end{equation}
where $\varphi(t)$ is a scalar function satisfying
\begin{equation}
  \label{eq:pen}
  \ddot \varphi = -\gamma \varphi^{k-1}. 
\end{equation}
This equation has first integral
\begin{equation}
  h(\varphi,\dot\varphi)=\dfrac{1}{2}\dot\varphi^2+\dfrac{\gamma}{k}\varphi^k.
\end{equation} 
Thus, in fact, we have to our disposal a family of particular
solutions parametrised by the value of first integral $h$.

Let $\vGamma_{\varepsilon}$ denote the phase curve of particular
solution \eqref{eq:party} for which
$h(\varphi,\dot\varphi)=\varepsilon$, i.e.,
\begin{equation}
  \vGamma_{\varepsilon} :=\defset{(\vq,\vp)\in\C^{2n}}{ (\vq,\vp)=(\varphi
    \vd,\dot \varphi\vd),\, h(\varphi,\dot\varphi)=\varepsilon}.
  \label{eq:phasecurve}
\end{equation}
The variational equation along solution~\eqref{eq:party} can be
written in the following form
\begin{equation}
  \label{eq:var2}
  \ddot \vx = \varphi(t)^{k-2}V''(\vd)\vx,
\end{equation}
where $V''(\vd)$ is the Hessian  of the potential $V$ calculated at a
Darboux point $\vd$.
If matrix $V''(\vd)$ is diagonalisable, then without loss of
generality we can assume that the above system splits into a direct
sum of second order equations
\begin{equation}
  \label{eq:svar2}
  \ddot x_i= \varphi(t)^{k-2} x_i, \qquad i=1,\ldots, n.
\end{equation}
It appears that the analysis of properties of these variational
equations gives very strong and computable obstructions to the
integrability of potential~$V(\vq)$.  In the case of proper Darboux
point these obstructions were found by J.~J.~Morales-Ruiz and
J.-P.~Ramis in~\cite{Morales:01::a}. Here we formulate them in the
following form. 
\begin{theorem}
  \label{thm:MoRa}
  Assume that the Hamiltonian system defined by Hamiltonian
  \eqref{eq:ham} with a homogeneous potential $V\in\C(\vq)$ of degree
  $k\in\Z\setminus\{-2,0,2\}$ satisfies the following conditions:
  \begin{enumerate}
  \item the potential has a proper Darboux point $\vd$ satisfying
    \[
    V'(\vd)=\gamma\vd,\quad\text{with}\quad
    \gamma\in\C^{\star}:=\C\setminus\{0\},
    \]
  \item matrix $\gamma^{-1}V''(\vd)$ is diagonalisable with
    eigenvalues $\lambda_1, \ldots, \lambda_n$;
  \item the system is integrable in the Liouville sense with first
    integrals which are meromorphic in a connected neighbourhood $U$
    of phase curve $\vGamma_{\varepsilon}$ with $\varepsilon\neq0$,
    and independent on $U\setminus \vGamma_{\varepsilon}$.
  \end{enumerate}
  Then each pair $(k,\lambda_i)$ for $i=1,\ldots,n$ belongs to an item
  of the following list
  \begin{equation}
    \label{tabMo}
    \begin{array}{clcl}
      1.& \left( k, p + \dfrac{k}{2}p(p-1)\right),&2.& \left(k,\dfrac 1
        {2}\left[\dfrac {k-1} {k}+p(p+1)k\right]\right), \\[1em]
      3.& \left(3,-\dfrac 1 {24}+\dfrac 1 {6}\left( 1 +3p\right)^2\right), &
      4.& \left(3,-\dfrac 1 {24}+\dfrac 3 {32}\left(  1  +4p\right)^2\right),
      \\[1em]
      5.& \left(3,-\dfrac 1 {24}+\dfrac 3 {50}\left(  1  +5p\right)^2\right),&
      6.& \left(3,-\dfrac 1 {24}+\dfrac{3}{50}\left(2 +5p\right)^2\right),\\[1em]
      7.&  \left(4,-\dfrac 1 8 +\dfrac{2}{9} \left( 1+ 3p\right)^2\right),&
      8.& \left(5,-\dfrac 9 {40}+\dfrac 5 {18}\left(1+ 3p\right)^2\right),\\[1em]
      9.& \left(5,-\dfrac 9 {40}+\dfrac 1 {10}\left(2+5p\right)^2\right),&
      10.&\left(-3,\dfrac {25} {24}-\dfrac 1 {6}\left(1+ 3p\right)^2\right) ,\\[1em]
      11.&\left(-3,\dfrac{25}{24}-\dfrac 3 {32}\left(1+4p\right)^2\right),&
      12.&\left(-3,\dfrac{25}{24}-\dfrac 3 {50}\left(1+5p\right)^2 \right),\\[1em]
      13.&\left(-3,\dfrac{25}{24}-\dfrac{3}{50}\left(2+5p\right)^2\right),&
      14.&\left(-4,\dfrac 9 8-\dfrac{2}{9}\left( 1+ 3p\right)^2\right),\\[1em]
      15.&\left(-5,\dfrac {49} {40}-\dfrac{5}{18}\left(1+3p\right)^2\right),&
      16.&\left(-5,\dfrac {49}{40}-\dfrac{1}{10}(2 +5p)^2\right)\\
    \end{array}
  \end{equation}
  where $p$ is an integer.
\end{theorem}
The proof of the above theorem is based on an analysis of the
differential Galois group of variational equations~\eqref{eq:svar2}. 
It appears that the existence of $n$ independent first
integrals which are in involution puts strong restrictions on this
group. Namely, its identity component must be Abelian, see
\cite{Morales:99::c}. This condition is translated into arithmetic
restrictions on eigenvalues of the Hessian $V''(\vd)$.
\begin{remark}
  \label{rem:yosan}
  The crucial  role in   the proof of the above theorem plays the
  following fact.  
Assuming that $k\neq 0$ and $\varepsilon\neq 0$, we make
transformation
\begin{equation}
  \label{eq:yo}
  t\longmapsto z:=\frac{1}{k\varepsilon} \varphi(t)^k.
\end{equation}
Then, $i$-th equation ~\eqref{eq:svar2} reads
\begin{equation}
  \label{eq:hyps}
  z(1-z)x'' + \left(\frac{k-1}{k} - \frac{3k-2}{2k}z\right)x' +
  \frac{\lambda_{i}}{2k} x=0 ,
\end{equation}
where prime denotes the differentiation with respect to $z$. It is
exactly the
Gauss hypergeometric equation
\begin{equation}
  z(1-z)x'' +[c-(a+b+1)z]x' -ab x=0,
  \label{eq:hyp}
\end{equation}
with parameters
\begin{equation}
  a+b=\frac{k-2}{2k},\qquad ab=-\frac{\lambda_{i}}{2k},\qquad c=1-\frac{1}{k}.
  \label{eq:abc}.
\end{equation}
For this equation the differential Galois group is perfectly known,
see, e.g., \cite{Kimura:69::}. 
We called transformation~\eqref{eq:yo} the Yoshida transformation as
it was found by  
  Haruo Yoshida  in~\cite{Yoshida:87::a}.
\end{remark}
For the excluded values of $k=\pm 2$ the variational equations do not
give any obstruction to the integrability. 

Case $k=0$ is special
because for this value of $k$ the phase curves corresponding to a
proper Darboux point have form different from that described by
\eqref{eq:phasecurve}. The necessary integrability conditions for this
case were found in \cite{mp:10::a}. They are following.
\begin{theorem}
\label{thm:gwe}
Assume  that  potential $V\in\C(\vq)$ 
homogeneous of degree $k=0$,  and that the following
conditions are satisfied:
\begin{enumerate}
 \item there exists a non-zero $\vd\in\C^n$ such that
   $V'(\vd)=\gamma\vd$ for a certain $\gamma\in\C^{\star}$; and
\item the potential is integrable in the Liouville sense with rational  first integrals.
\end{enumerate}
Then:
\begin{enumerate}
\item all eigenvalues of $\gamma^{-1}V''(\vd)$ are integers; and
 \item the matrix $V''(\vd)$ is diagonalizable.
\end{enumerate}
\end{theorem}
The fact that $\vd$ is a Darboux point and $V$ is a homogeneous
function of degree $k$ implies that $\vd$ is an eigenvector of matrix
$\gamma^{-1}V''(\vd)$ with the corresponding eigenvalue $\lambda_n=k-1$.  This
eigenvalue belongs to the first item of table~\eqref{tabMo} and so it
does not give any restriction to the integrability. Later we will
called it the trivial eigenvalue of $V''(\vd)$.

For a fixed $k\in\Z\setminus\{-2,0,2\}$, the subset of rational
numbers $\lambda\in\Q$ such that the pair $(k,\lambda)$ belongs to an
item of the Morales-Ramis table is denoted by $\scM_k$. By
Theorem~\ref{thm:gwe}, we have also $\scM_0:=\Z$.  For the later
use we define also the set
\begin{equation}
  \scI_k=\{\Lambda\in\Q\ |\ \Lambda+1\in\scM_k\}.
\end{equation} 

Note that for an arbitrary $k\in\Z\setminus\{-2,0,2\}$ sets $\scM_k$
and $\scI_k$ are not finite.
\begin{remark}
  \label{rem:nd}
  As it was explained in \cite{Maciejewski:09::} the assumption that
  $V''(\vd)$ is diagonalisable is irrelevant.  In a case when
  $V''(\vd)$ is not diagonalisable the necessary conditions for the
  integrability are exactly the same: if the potential is integrable, then
  each eigenvalue $\lambda$ of $V''(\vd)$  belongs  to an item of
  table~\eqref{tabMo}. Furthermore, the presence of a Jordan block of dimensions greater than 1
  gives additional integrability obstructions.  Namely, if the Jordan
  form of $V''(\vd)$ has a block of dimension greater that two 
then the system is not integrable. Moreover, if it  has a two dimensional block  and the corresponding $\lambda$
  belongs to the first item of table~\eqref{tabMo}, then the system is
  also not integrable. This facts  were  proved in~\cite{Maciejewski:09::}.
\end{remark}
Necessary conditions for the integrability can be deduced also from an
analysis of the differential Galois groups of variational equations
along a particular solution given by an improper Darboux point. For
polynomial potentials it was done in \cite{mp:09::a}. Here we
generalise this result to the class of rational potentials.
\begin{theorem}
  \label{th:nonint}
  Assume that homogeneous potential $V\in \C(\vq)$ of degree
  $k\in\Z\setminus \{-2,0,2\}$ possesses an improper Darboux point
  $\vd$. If $V$ is integrable in Liouville sense with rational first
  integrals, then all eigenvalues of $V''(\vd)$ vanish.
\end{theorem}

In order to find additional integrability obstructions we restrict our
considerations to potentials with two degrees of freedom. In this case
for a proper Darboux point $\vd$ there is only one non-trivial
eigenvalue of Hessian $ V''(\vd)$. Its value is given by
\begin{equation*}
  \lambda=\gamma^{-1}\tr V''(\vd)-(k-1). 
\end{equation*} 
We say that $\vd$ is a simple Darboux point iff $\lambda\neq 1$. Later
we will give more algebraic definition of simplicity.

For an improper Darboux point $\vd$ we have also only one non-trivial
eigenvalue of the form
\begin{equation*}
  \lambda=\tr V''(\vd).
\end{equation*}
The trivial eigenvalue of $ V''(\vd)$ for an improper Darboux point
$\vd$ is $\lambda_2=0$.

In further part of this paper we show that for a homogeneous rational
potential of degree $k=r-s$, if number of Darboux points is finite,
then there exists $l\leq r+s$ proper Darboux points. For each proper
Darboux point $\vd_i$ we calculate the ``shifted'' non-trivial
eigenvalue
\begin{equation}
  \Lambda_i=\lambda_i-1=\gamma^{-1}\tr V''(\vd_i)-k,
  \label{eq:bigL}
\end{equation}
and we set $\bLambda=(\Lambda_i,\ldots,\Lambda_{l})\in\C^{l}$.  It
appears that $\bLambda$  cannot be an arbitrary element of $\C^l$. Namely, there
is a
certain universal relation between them.
\begin{theorem}
  \label{thm:rel}
  Assume that a rational homogeneous potential $V(q_1,q_2)$ of degree
  $k\in\Z$ satisfies three conditions:
  \begin{description}
  \item[C1] it has $0<l\leq r+s$ proper Darboux points and all of them
    are simple;
  \item [C2] $U$ is not factorisable neither by $(q_2+\rmi q_1)$, nor
    by $(q_2-\rmi q_1)$;
  \item [C3] if $W$ is factorisable by $(q_2+ \rmi q_1)$, or by $(q_2-
    \rmi q_1)$ then multiplicity of these factor is one.
  \end{description}
  Then the shifted non-trivial eigenvalues $\Lambda_i$ given by
  \eqref{eq:bigL} satisfy the relation
  \begin{equation}
    \sum_{i=1}^{l}\dfrac{1}{\Lambda_i}=-1.
    \label{eq:relka1}
  \end{equation} 
\end{theorem}

In order to formulate a generalisation of the above theorem we have to
introduce some definitions. As we will explain later a non-constant
rational homogeneous function $F(q_1,q_2)$ can be represented
uniquely, up to multiplicative constant, in the following form
\[
F(q_1,q_2)=\prod_{i=1}^rl_i(q_1,q_2)^{e_i},
\]
where $r\in\N$, $e_i\in\Z^{\star}$ and $l_i(q_1,q_2)$ are homogeneous
polynomials of degree $1$. We say that $l_i(q_1,q_2)$ with $1\leq
i\leq r$ is a factor of $F(q_1,q_2)$ iff $e_i>0$.  In this case
integer $e_i$ is called the multiplicity of factor $l_i(q_1,q_2)$. Let
$r_{\pm}$ and $s_{\pm}$ be the respective multiplicities of linear
factors $(q_2\pm\rmi q_1)$ of $W$ and $U$, respectively.  We also
define  the following step function
\begin{equation}
  \theta_{x,y}:=\begin{cases}
    0 &\text{for}\quad x<y,   \\
    1 &\text{for}\quad x\geq y, 
  \end{cases}
  \label{eq:schodkowa}
\end{equation}
where $x,y\in\R$.
\begin{theorem}
  \label{thm:rell2}
  Assume that a rational homogeneous potential $V(q_1,q_2)$ of degree
  $k=r-s\in\Z$ satisfies three conditions:
  \begin{description}
  \item[C1] it has $0<l\leq r+s$ proper Darboux points and all of them
    are simple;
  \item [C2] neither $r_{+}$, nor $r_{-}$ is equal to $k/2$;
  \item [C3] neither $s_{+}$, nor $s_{-}$ is equal to $-k/2$.
  \end{description}
  Then the shifted non-trivial eigenvalues given in \eqref{eq:bigL}
  satisfy the relation
  \begin{equation}
    \label{eq:relka2}
    \begin{split}
      \sum_{i=1}^l\frac{1}{\Lambda_i} = -1 &- \theta_{r_+,2}
      \frac{r_+}{k-2r_+} - \theta_{r_-,2} \frac{r_-}{k-2r_-} +
      \theta_{s_+,1} \frac{s_+}{k+2s_+} + \theta_{s_-,1}
      \frac{s_-}{k+2s_-}.
    \end{split}
  \end{equation}
\end{theorem}

The importance of relations \eqref{eq:relka1} and \eqref{eq:relka2}
 follows from the following two theorems.
\begin{theorem}
  \label{thm:finiteness}
  Let us consider \eqref{eq:relka1} or \eqref{eq:relka2} as an
  equation for
  $(\Lambda_1,\ldots,\Lambda_{l})\in\underbrace{\C^{\star}\times\cdots\times
    \C^{\star}}_{l}$.  Then, for $k\in\Z\setminus\{-2,2\}$, it has at most a finite number of
  solutions contained in $\underbrace{\scI_k\times\cdots\times
    \scI_k}_{l}$.
\end{theorem}
In other words, this theorems says that if the potential of degree
$k\in\Z\setminus\{-2,0,2\}$ is integrable, then we have only a finite
number of choices for eigenvalues of  $V''(\vd)$.

The plan of this paper is the following. In Section~\ref{DarPoints} we
describe Darboux points of rational potentials as elements of complex
projective space $\mathbb{CP}^{n-1}$ using homogeneous as well as
affine coordinates and we show some their general properties.  In
Section~\ref{sec:improp} we prove Theorem~\ref{th:nonint}.  The
next Section~\ref{sec:two} contains more precise description of
properties of Darboux points for homogeneous rational potentials with
two degrees of freedom.  In particular, in Subsection ~\ref{sec:darb2}
we determine the number of possible proper Darboux points,
characterise potentials with multiple proper and improper Darboux
points. In the end of this subsection we determine potentials with
infinitely many proper Darboux points as well as without proper
Darboux points. In Subsection~\ref{NFI} we prove
Theorems~\ref{thm:rel} and~\ref{thm:rell2} which give us universal
relations between all eigenvalues of Hessian $V''(\vd)$. The importance of such
relations is presented in Subsection
\ref{sec:finiteness}, where the proof of Theorem~\ref{thm:finiteness}
is presented.  This
theorem is the bases of the classification programme of integrable
rational potentials.

In Section~\ref{appl} we give some applications of the presented
theory. In Subsection~\ref{k-} we introduced a class of potentials
with a negative degree of homogeneity which satisfy all necessary
integrability. We obtain it by reconstruction from the fixed
non-trivial eigenvalues calculated at proper Darboux points assuming
that their number is maximal and all these eigenvalues are the same
and belong to first family from table~\ref{tabMo}. In
Subsection~\ref{dorizzi} we distinguish special sub-class of these
potentials which have additional first integrals.  In the end of
this paper in Subsection~\ref{k+} we made analogical analysis like in
Subsection~\ref{k-} but for homogeneous rational potential of positive
degree of homogeneity. 

\section{General properties of Darboux points of rational potentials}
\label{DarPoints}

Here we will use some standard notions of algebraic geometry. Thus,
for a collection of polynomials $F_1,\ldots,F_s\in\C[\vq]$, we denote
by $\scV(F_1,\ldots,F_s)$ the set of their common zeros, i.e.,
\[
\scV(F_1,\ldots,F_s)=\left\{\vq\in\C^n\ |\
  F_1(\vq)=\cdots=F_s(\vq)=0\right\}.
\]
Such a subset of $\C^n$ is called an algebraic set.

Let us note that a rational potential
\begin{equation}
  V(\vq)=\dfrac{W(\vq)}{U(\vq)}
  \label{eq:ratek}
\end{equation} 
cannot be considered as a function which is well defined for all
$\vq\in\C^n$.  In fact, if polynomials $W(\vq)$ and $U(\vq)$ are
relatively prime and $\deg U(\vq)>0$, then $V(\vq)$ is not well
defined whenever $U(\vq)=0$. At such a point value of $V(\vq)$ is
either infinity or is indefinite. The second case occurs at points
$\vq\in\C$ such that $U(\vq)=W(\vq)=0$. The set of all such points of
potential $V$ we denote by $\scN(V)$. Thus $\scN(V)=\scV(U,W)$. The
set of poles $\scP(V)$ of potential $V$ is the set of points
$\vq\in\C$ at which the values of $V$ are infinite. Thus
$\scP(V)=\scV(U)\setminus\scN(V)$.

We show that for $n=2$ we have $\scN(V)=\{0\}$. In fact, we know that
every homogeneous complex polynomial in two variables is a product of
linear forms.  Thus we can write
\[
\begin{split}
  W(q_1,q_2)&=\prod_{i=1}^r\left(\alpha_1^{(i)}q_1+\alpha_2^{(i)}q_2
  \right),\quad |\alpha_1^{(i)}|^2+|\alpha_2^{(i)}|^2\neq0,\quad
  \text{for}\quad
  i=1,\ldots,r,\\
  U(q_1,q_2)&=\prod_{j=1}^s\left(\beta_1^{(j)}q_1+\beta_2^{(j)}q_2
  \right),\quad |\beta_1^{(j)}|^2+|\beta_2^{(j)}|^2\neq0,\quad
  \text{for}\quad j=1,\ldots,s.
\end{split}
\]
If $W(q_1,q_2)=U(q_1,q_2)=0$ for a certain $(q_1,q_2)$, then

\[
\left(\alpha_1^{(i)}q_1+\alpha_2^{(i)}q_2\right)=\left(\beta_1^{(j)}
  q_1+\beta_2^ { (j) } q_2 \right)=0,
\]
for a certain $i\in\{1,\ldots,r\}$ and $j\in\{1,\ldots,s\}$.  If
$(q_1,q_2)\neq(0,0)$, then
\[
\beta_1^{(j)}=\gamma\alpha_1^{(i)},\quad\text{and}\quad
\beta_2^{(j)}=\gamma\alpha_2^{(i)},
\]
for a certain $\gamma\in\C$.  However, it means that $U$ and $W$ have
a common factor $\left(\alpha_1^{(i)}q_1+\alpha_2^{(i)}q_2\right)$ and
this is a contradiction with our assumption that $U$ and $W$ are
relatively prime.

In order to give a precise definition of Darboux points we consider
rational vector $V'(\vq)\in\C(\vq)^n$. Although we assumed that
$V(\vq)=W(\vq)/U(\vq)$ with relatively prime polynomials $W$ and $U$,
it does not imply that polynomials $U$ and
\begin{equation}
  S_i=\dfrac{\partial W}{\partial q_i}U-W\dfrac{\partial U}{\partial
    q_i},\qquad i=1,\ldots,n,
\end{equation} 
are relatively prime. In other words writing
\[
\partial_iV=\dfrac{\partial V}{\partial q_i}=\dfrac{S_i}{U^2},
\]
we cannot be sure that the sets of indefiniteness of $\partial_iV$ is
$\scV(S_i,U)$.

From definition \eqref{eq:dar2} it follows that if $\vd$ is a Darboux
point of potential $V$, then $\alpha\vd$, where $\alpha\in\C^{\star}$,
is also a Darboux point which gives the same phase curves
\eqref{eq:phasecurve}. This is why it is convenient to consider
Darboux points as points in projective space. Thus we consider a
non-zero $\vq\in\C^n$ as the homogeneous coordinates of a point in
$\mathbb{CP}^{n-1}$ and we write
$[\vq]:=[q_1:\cdots:q_n]\in\mathbb{CP}^{n-1}$.

In order to avoid formalities we say that $[\vd]\in\mathbb{CP}^{n-1}$
is a Darboux point of the potential $V$ iff
\begin{enumerate}
\item components of $V'(\vq)$ are well defined at $\vq=\vd$ and ;
\item $V(\vq)$ is well defined at $\vq=\vd$ and;
\item $\vd\wedge V'(\vd)=0$.
\end{enumerate}

The set of all Darboux points of potential $V$ we denote by
$\scD(V)$. A Darboux point $[\vd]\in\scD(V)$ is called a proper
Darboux point iff $V'(\vd)\neq0$, otherwise it is called an improper
Darboux point. The set of all proper Darboux points of potential $V$
is denoted by ${\scD}^\star(V)$.  We say that $[\vd]\in\scD(V)$ is an
isotropic Darboux point iff
\[
d_1^2+\cdots+d_n^2=0.
\]
The set of all isotropic Darboux points of potential $V$ is denoted by
$\scD_0(V)$.

As it was shown in \cite{mp:09::a} for a polynomial potential $V$ set
$\scD(V)$ is an algebraic subset of $\mathbb{CP}^{n-1}$. For a
rational $V$ it is generally not the case. However, we can consider
the following homogeneous polynomials
\begin{equation}
  R_{i,j}=q_i\left(\dfrac{\partial W}{\partial q_j}U-\dfrac{\partial U}{\partial
      q_j}W\right)-q_j\left(\dfrac{\partial W}{\partial q_i}U-\dfrac{\partial
      U}{\partial q_i}W\right),\quad 1\leq i<j\leq n.
\end{equation} 
Then the algebraic set
\begin{equation}
  \scR(V)=\scV(R_{1,2},\ldots,R_{n-1,n})\subset \mathbb{CP}^{n-1},
\end{equation} 
contains all Darboux points i.e. $\scD(V)\subset \scR(V)$, and clearly
$\scN(V)\subset \scR(V)$.

It appears that it is convenient to perform some calculations in
affine coordinates on $\mathbb{CP}^{n-1}$. They are introduced in the
following way.  A point in the $(n-1)$ dimensional complex projective
space $\CP^{n-1}$ is specified by its homogeneous coordinates
$[\vq]=[q_1: \cdots:q_{n}]$, where $\vq=(q_1,\ldots,q_n)\in\C^{n}$. 
We define $n$ open subsets $U_i$ of $\CP^{n-1}$
\begin{equation}
  U_i:=\defset{[q_1: \cdots:q_n]\in\CP^{n-1}}{q_i\neq 0}\mtext{for}i=1,\ldots,n.
\end{equation} 
Clearly
\begin{equation}
  \CP^{n-1}=\bigcup_{i=1}^n U_i,
\end{equation} 
and we have natural coordinate maps
\begin{equation*}
  \theta_i:\CP^{n-1} \supset U_i\rightarrow \C^{n-1} ,\qquad
  \theta_i([\vq])=(x_1,\ldots,x_{n-1}),
\end{equation*} 
where
\begin{equation}
  (x_1,\ldots,x_{n-1})=\left(\frac{q_1}{q_i}, \ldots, 
    \frac{q_{i-1}}{q_i}, \frac{q_{i+1}}{q_i}, \ldots, \frac{q_{n}}{q_i}\right). 
\end{equation} 
Each $U_i$ is homeomorphic to $\C^{n-1}$. It is easy to check that
charts $(U_i,\theta_i)$, $i=1,\ldots, n$ form an atlas which makes
$\CP^{n-1}$ an holomorphic $(n-1)$-dimensional manifold.  It is
customary to choose one $U_i$, e.g., $U_1$, and call it the affine
part of $\CP^{n-1}$ and coordinates on them we call affine
coordinates.

It means that on $U_1:= \defset{[\vq]\in\CP^{n-1}}{q_1 \neq0}$ we
define affine coordinates
\begin{equation}
  \label{eq:th1}
  \theta_1:U_1\rightarrow \C^{n-1}, \quad \vx:=(x_1,\ldots,x_{n-1})=
  \theta_1([\vq]),
\end{equation}
by
\begin{equation}
  \label{eq:xi}
  x_i=\frac{q_{i+1}}{q_1}, \mtext{for}i=1,\ldots, n-1.
\end{equation}

For a homogeneous polynomial $F\in\C[\vq]$ we define its
dehomogenisation $f\in\C[x_1,\ldots,x_{n-1}]$ in the following way
\[
f(x_1,\ldots,x_{n-1}):=F(1,x_1,\ldots,x_{n-1}).
\]

Now we prove the following lemma.
\begin{lemma}
  \label{lem:aff}
  \begin{equation}
    \label{eq:dvu1}
    \theta_1(\scR(V)\cap U_1)=\scV(g_1,\dots,g_{n-1}),
  \end{equation}
  where polynomials $g_1,\ldots,g_{n-1}\in\C[ \vx]$ are given by
  \begin{equation}
    \label{eq:g0}
    g_0:= k uw -\sum_{i=1}^{n-1}
    x_i\left (u\pder{w}{x_i}-w\pder{u}{x_i}\right),  \qquad k=r-s,
  \end{equation}
  and
  \begin{equation}
    \label{eq:gi}
    g_i := u\pder{w}{x_i}- w\pder{v}{x_i}-x_i g_0, \mtext{for}i=1, \ldots, n-1.
  \end{equation}
  Moreover, if $[\vd]\in \scR(V)\cap U_1$ is an improper Darboux, then
  its affine coordinates $ \va := \theta_1([\vd])\in\C^{n-1}$
  satisfy $g_0(\va)=0$.
\end{lemma}
\begin{proof}
  Darboux points are determined by common zeros of polynomials
  $R_{ij}$. We calculate their  dehomogenisations $r_{ij}$.

  At first we express partial derivatives of $W$ and $U$ in affine
  coordinates
  \begin{equation}
    \begin{split}
      &\dfrac{\partial W}{\partial
        q_1}(\vq)=q_1^{r-1}\left[rw-\sum_{j=1}^{n-1}x_j\dfrac{\partial
          w}{\partial x_j}(\vx)\right],\qquad \dfrac{\partial
        W}{\partial q_i}(\vq)=q_1^{r-1}\dfrac{\partial w}{\partial
        x_{i-1}}(\vx),\\
      &\dfrac{\partial U}{\partial
        q_1}(\vq)=q_1^{s-1}\left[su-\sum_{j=1}^{n-1}x_j\dfrac{\partial
          u}{\partial x_j}(\vx)\right],\qquad \dfrac{\partial
        U}{\partial q_i}(\vq)=q_1^{s-1}\dfrac{\partial u}{\partial
        x_{i-1}}(\vx),
    \end{split}
  \end{equation} 
  for $i=2,\ldots,n-1$. As
  \[
  r_{i,j}(x_1,\ldots,x_{n-1}):=R_{i,j}(1,x_1,\ldots,x_{n-1}),\mtext{for}
  1\leq i<j\leq n-1,
  \]
  through direct calculations we obtain 
  \begin{equation}
    \begin{split}
      g_i:= &r_{1,i}=\dfrac{\partial u(\vx)}{\partial
        x_{i-1}}w(\vx)-\dfrac{\partial w(\vx)}{\partial
        x_{i-1}}u(\vx)+\\ &x_{i-1}\left[
        kuw+
\sum_{j=1}^{n-1}x_j\left(w(\vx)\dfrac{\partial
            u(\vx)}{\partial x_j}-u(\vx)\dfrac{\partial
            w(\vx)}{\partial x_j}
        \right)\right]
    \end{split}
  \end{equation}
for $1\leq i\leq n-1$. Then we easily get 
\begin{equation}
     r_{i+1,j+1}=x_ig_j-x_jg_i,\qquad \mtext{for} 1\leq i<j\leq n-1.
  \end{equation}
  Thus clearly
  $\scV(r_{1,2},\ldots,r_{n-1,n})=\scV(g_1,\dots,g_{n-1})$.

  At an improper Darboux point $[\vd]$ we have
  \[
  V'(\vd)=\dfrac{1}{U(\vd)^2}\left(W'(\vd)U(\vd)-U'(\vd)W(\vd)
  \right)=0.
  \]
  By assumptions, $U(\vd)\neq 0$, so
  \[
  W'(\vd)U(\vd)-U'(\vd)W(\vd)=0.
  \]
  Passing into affine coordinates we obtain
  \[
  \begin{split}
    & \dfrac{\partial W}{\partial q_1}U-\dfrac{\partial U}{\partial
      q_1}W=q_1^{r+s-1}g_0(\vx),\\
    &\dfrac{\partial W}{\partial q_j}U-\dfrac{\partial U}{\partial
      q_j}W=q_1^{r+s-1}\left(\dfrac{\partial w}{\partial
        x_{j-1}}u-\dfrac{\partial u}{\partial
        x_{j-1}}w\right)=q_1^{r+s-1}(g_{j-1}(\vx)+x_{j-1}g_0(\vx)),
  \end{split}
  \]
  for $j=2,\ldots,n$. Thus, in fact at improper Darboux point we have
  $g_{0}(\vx)=0$, and this ends the proof.
\end{proof}
\section{Proof of Theorem~\ref{th:nonint}}
\label{sec:improp}

The variational equations along particular solution
\begin{equation*}
  t\mapsto (\varphi (t)\vd,\dot\varphi (t)\vd)
\end{equation*} 
have the form
\begin{equation}
  \ddot\vx=-\varphi(t)^{k-2}V''(\vd)\vx.
  \label{eq:wariat}
\end{equation} 
We make a linear change of variables $\vx=\vC\veta$ which transforms
\eqref{eq:wariat} into
\begin{equation}
  \ddot\veta=-\varphi(t)^{k-2}\vD\veta,\qquad \vD=\vC^{-1}V''(\vd)\vC.
  \label{eq:wariatjor}
\end{equation} 
Hence we can choose $\vC$ in such a way that $\vD$ is the Jordan form
of $V''(\vd)$. System \eqref{eq:wariatjor} contains as a subsystem the
following direct product of equations
\begin{equation*}
  \ddot\eta_i=-\lambda_i\varphi(t)^{k-2}\eta_i,\qquad i=1,\ldots,m\leq n,
\end{equation*} 
where $\lambda_1,\ldots,\lambda_m$ are pairwise different eigenvalues
of $V''(\vd)$.

We have to prove that if $\lambda_i\neq0$ for a certain $1\leq i\leq
m$, then the identity component of differential Galois group of system
\eqref{eq:wariatjor} is not Abelian. To this end it is enough to show
that the identity component of differential Galois group of a single
equation
\begin{equation}
  \ddot\eta_i=-\lambda_i\varphi(t)^{k-2}\eta_i,\quad\text{with}\quad
  \lambda_i\neq0,
  \label{eq:niezerowe}
\end{equation}
is not Abelian for an appropriate choice of $\varphi$. For an improper
Darboux point $\varphi(t)$ satisfies $\ddot\varphi=0$. So, we can take
$\varphi(t)=At$, where $A\in\C$ and we choose $A$ such that
$-\lambda_i(At)^{k-2}=t^{k-2}$.

Denoting $\eta:=\eta_i$ and $\alpha=k-2$ we rewrite
\eqref{eq:niezerowe} as
\begin{equation}
  \label{war}
  \ddot \eta=t^{\alpha}\eta.
\end{equation} 

Now we prove the following lemma
\begin{lemma}
  For $\alpha\in\Z\setminus\{-4,-2,0\}$ the differential Galois group
  of equation \eqref{war} is $\mathrm{SL}(2,\C)$.
\end{lemma}
\begin{proof}
  For $\alpha\in\N$ the statement of lemma is proved in
  \cite{mp:09::a}, see proof of Theorem 2.4 in \cite{mp:09::a}. Hence,
  we assume that $\alpha=-\beta\leq 0$.  Let us introduce new
  independent variable $\tau=1/t$. Equation \eqref{war} transforms
  into
  \[
  \eta''+\dfrac{2}{\tau}\eta'-s^{\beta-4}\eta=0,
  \]
  where prime denotes the differentiation with respect to $\tau$. Now,
  introducing new dependent variable $\xi=\tau\eta$ we obtain
  \begin{equation}
    \xi''=\tau^{\beta-4}\xi.
    \label{eq:szajs}
  \end{equation} 
  The transformation that we made, does not change the differential
  Galois group of the considered equation. Hence, our lemma is already
  proved for $\beta>4$, that is, for $\alpha<-4$. We have only to
  prove our lemma for $\alpha=-1$, and for $\alpha=-3$.

  We prove this with the help of the Kovacic algorithm, for details,
  see \cite{Kovacic:86::}. This algorithm allows to decide whether the
  considered second order linear equation with rational coefficients
  has only Liouvillian solution. At the same time it provides detailed
  information about the differential Galois group of this equation. In
  particular, if equation does not admit a Liouvillian solution, then
  its differential Galois group is $\mathrm{SL}(2,\C)$.

  We apply the Kovacic algorithm in its original formulation
  \cite{Kovacic:86::}.  All definitions and notation is exactly the
  same as in this paper.

  The algorithm is divided into mutually exclusive cases. Only in
  first three cases the equation admits a Liouvillian solution. If
  none of these three cases occurs, then the differential Galois group
  of the considered equation is $\mathrm{SL}(2,\C)$. In this case the
  whole group coincides with  its identity
  component, so both are not Abelian.

  We show that for $\alpha=-1$ none of first three cases of the
  Kovacic algorithm is allowed. To this end we apply Theorem on page 8
  of \cite{Kovacic:86::} which gives the necessary conditions for the
  respective cases. Equation \eqref{war} has two singularities: a pole
  of order $-\alpha$ at $t=0$ and the infinity which has also order
  $-\alpha$.

  The necessary condition for case I is that the order of infinity is
  either even or greater than 2. For $\alpha=-1$ it is not the case.

  Case II cannot occur because the necessary condition for it is that
  there is at least one pole of odd order greater than 2 or else of
  order two.

  Case III cannot also occurs  because the necessary condition for it
  is that the infinity has order at least 2.

  As result, we showed that for $\alpha=-1$ equation \eqref{war} has
  no Liouvillian solution and its differential Galois group is
  $\mathrm{SL}(2,\C)$.

  Using the same theorem for $\alpha=-3$ we easily notice that if
  equation \eqref{war} has a Liouvillian solution, then we fall into
  case II of the algorithm. In this case exponents at the
  singularities are
  \[
  E_0=\{3\}\qquad \text{and}\qquad E_{\infty}=\{0,2,4\},
  \]
  According to the algorithm if this case occurs, then there exist
  $e_0\in E_0$ and $e_{\infty}\in E_{\infty}$ such that number
  \[
  d=\frac{1}{2}(e_{\infty}-e_0)
  \]
  is a non-negative integer. But it is impossible because $e_0$ is odd
  and $e_{\infty}$ even.  Thus also for $\alpha=-3$ differential
  Galois group of \eqref{war} is $\operatorname{SL}(2,\C)$.
\end{proof}
Let us notice that if   $\alpha=-4$, then  the  degree of homogeneity
of the potential  $k=-2$. In this case  equation \eqref{war} has
two Liouvillian solutions
\[
\eta_1=te^{\frac{\sqrt{\lambda}}{t}},\qquad
\eta_2=te^{-\frac{\sqrt{\lambda}}{t}}
\]
satisfying the following relations
\[
\dfrac{\eta_1'}{\eta_1}=\dfrac{t-\sqrt{\lambda}}{t^2}\in\C(t),\qquad
\dfrac{\eta_2'}{\eta_2}=\dfrac{t+\sqrt{\lambda}}{t^2}\in\C(t),\qquad
\eta_1\eta_2=t^2\in\C(t).
\]

For $\alpha=-2$, i.e., $k=0$, the  variational equation  has solutions
\[
\eta_1=t^{\frac{1}{2}+\sqrt{1+4\lambda}},\qquad
\eta_2=t^{\frac{1}{2}-\sqrt{1+4\lambda}},
\]
 satisfying the following relations
\[
\dfrac{\eta_1'}{\eta_1}=\dfrac{1+\sqrt{1+4\lambda}}{2t}\in\C(t),\qquad
\dfrac{\eta_2'}{\eta_2}=\dfrac{1-\sqrt{1+4\lambda}}{2t}\in\C(t),\qquad
\eta_1\eta_2=t\in\C(t).
\]
The above shows that in cases with $k=-2$ and $k=0$ the differential Galois
of variational equations is Abelian as it is 
subgroup of diagonal group. It  is an
open problem  if the presence of a Jordan block gives additional 
 obstructions for the integrability  in these cases. For $\alpha=0$, corresponding
to $k=2$, variational equation \eqref{war} becomes an equation with
constant coefficient, thus it does not give any integrability
obstructions.

\section{Two degrees of freedom}
\label{sec:two}
\subsection{More about properties of Darboux points}
\label{sec:darb2}
In this section we assume that $n=2$. A rational homogeneous potential
can be written in the following form
\begin{equation}
  \label{pot2zm}
  \begin{split}
    &V(\vq)=\dfrac{W(\vq)}{U(\vq)},\quad \text{where}\\
    &W(\vq)=\sum_{i=0}^rw_{r-i}q_1^{r-i}q_2^i,\quad
    U(\vq)=\sum_{j=0}^su_{s-j}q_1^{ s-j } q_2^j.
  \end{split}
\end{equation} 
We assume that polynomials $W$ and $U$ are relatively prime.  Since
polynomials $W$ and $U$ are homogeneous of degrees $r$ and $s$,
respectively, potential $V$ is a homogeneous function of degree
$k=r-s$. 

We denote by $z=q_2/q_1$ the affine coordinate on $U_1\subset
\mathbb{CP}^1$.  Then, dehomogenizations $w(z)$, $u(z)$ and $v(z)$ of
polynomials $W(\vq)$, $U(\vq)$ and potential $V(\vq)$ have the
following forms
\[
w(z)=\sum_{i=0}^rw_{r-i}z^i,\qquad u(z)=\sum_{j=0}^su_{s-j}z^j,\qquad
v(z)=\dfrac{w(z)}{u(z)}.
\]
On $U_2\subset \mathbb{CP}^1$ as a coordinate we take
$\zeta=1/z$. Now, dehomogenizations of $U$, $V$ and $W$ are following
\begin{equation}
  \label{eq:wtuvw}
  \widetilde
  u(\zeta):=U(\zeta,1)=\zeta^su(1/\zeta)=\sum_{i=1}^su_i\zeta^i, \qquad
  \widetilde v(\zeta):=\zeta^kv(1/\zeta), \qquad \widetilde w(\zeta):=\zeta^rw(1/\zeta).
\end{equation}

We can also write
\begin{equation}
  \label{eq:wuab}
  w(z)=\gamma\prod_{i=1}^l(z-a_i)^{\alpha_i}, \mtext{and} u(z)=\prod_{i-1}^m(z-b_i)^{\beta_i},
\end{equation}
where $\alpha_i, \beta_i\in\N$. We assume that for $w(z)=\gamma$ for
$l=0$, and $u(z)=1$ for $m=0$.

For $n=2$ Darboux points of potential $V$ are zeros of one rational
function
\[
F(\vq)=\vq\wedge V'(\vq)=q_1\dfrac{\partial V}{\partial
  q_2}-q_2\dfrac{\partial V}{\partial q_1}.
\]
Its dehomogenisation $f(z)$ is following
\begin{equation}
  \label{eq:fz}
  f(z)=F(1,z)=(1+z^2)v'(z)-kzv(z).
\end{equation} 

Now, polynomials $g:=g_1$ and $h:=g_0$, which are defined in
Lemma~\ref{lem:aff}, have the forms
\begin{equation}
  \label{funGH}
  \begin{split}
    h(z)&=kw(z)u(z)-z\left[w'(z)u(z)-u'(z)w(z)\right],\\
    g(z)&=(1+z^2)\left[w'(z)u(z)-u'(z)w(z)\right]-kzw(z)u(z).
  \end{split}
\end{equation} 
Notice that $g=fu^2$.

For an investigation of equivalent potentials in affine coordinates we
introduce the following notation. Let $A=\left[\begin{smallmatrix}a &
  b\\-b & a
\end{smallmatrix}\right]
$, with $a^2+b^2\neq 0$ be an element of $\mathrm{PO}(2,\C)$, and
$R(\vq)$ a rational homogeneous function of degree $p$. We define
$R_{A}(\vq):=R(A\vq)$. Let $r(z)$ and $\widetilde r(\zeta)$ be
dehomogenizations of $R(\vq)$. Then the dehomogenizations $r_{A}(z)$
and $\widetilde r_{A}(\zeta)$ are given by
\begin{equation}
  \label{eq:ra}
  r_{A}(z)=\left( a+b z \right)^p r(\tau_A(z)), \qquad
  \widetilde r_{A}(\zeta)=\left( a -b \zeta \right)^p r(\tau_{A}^{-1}(\zeta)), 
\end{equation}
where
\begin{equation*}
  \tau_A(z)=\frac{az-b}{bz+a}, \qquad \tau_A^{-1}(z)=\frac{az+b}{-bz+a}
\end{equation*}

In order to avoid numerous repetitions we assume implicitly in this
section that polynomials $U$ and $W$ are relatively prime. Moreover,
in order to simplify formulations instead of saying that $z\in\C$ is
an affine coordinate of a Darboux point we just say that $z$ is a
Darboux point. Taking into account this convention the sets of all
Darboux points, proper and improper Darboux points of potential $V$
which are contained in the affine part of $\CP^1$ can be defined in
the following way
\begin{equation}
  \begin{split}
    \label{zb1}
    \scD(V)&:=\left\lbrace z\in\C\ | \ g(z)=0 \ \mtext{and} \ u(z)\neq
      0 \right\rbrace,
    \\
    \scD^{\star}(V)&:=\left\lbrace z\in \C \ | \ g(z)=0 \ \text{and} \
      h(z)\neq 0 \ \text{and} \ u(z)\neq 0 \right\rbrace,
    \\
    \scD(V)\setminus\scD^{\star}(V)&:=\left\lbrace z\in \C \ | \
      g(z)=0 \ \text{and} \ h(z)= 0 \ \text{and} \ u(z)\neq
      0\right\rbrace,
  \end{split}
\end{equation}  
respectively.  We also define the set of isotropic Darboux points like
\begin{equation}
  \label{zb2}
  \scD_0(V)=\left\lbrace z\in \left\{-\rmi,\rmi\right\}\ | \ g(z)=0 \ \text{and}
    \  u(z)\neq
    0 \right\rbrace .
\end{equation} 
If polynomial $g$ is not identically zero, then the number of Darboux points of
the potential is bounded by its degree. Just direct calculations show
the following.

\begin{proposition}
  \label{pro:maxd}
  The degree of polynomial $g$ is not greater than $r+s$. Moreover,
  $\deg g=r+s$ if and only if
  \begin{equation}
    \label{eq:maxdeg}
    w_1u_0-w_0u_1\neq 0.
  \end{equation}
\end{proposition}
It is worth to notice that condition
\begin{equation}
  \label{eq:max0}
  w_1u_0-w_0u_1=0,
\end{equation}
means that polynomial
\begin{equation*}
  S_1=\pder{W}{q_1}U -W\pder{U}{q_1},
\end{equation*}
vanishes at points $(0,q_2)$ with $q_2\in\C^{\star}$. Thus, if
$U(0,1)\neq0$, then $[0:1]$ is a Darboux point of $V$.

We show that only radial potentials have infinite number of Darboux
points.
\begin{proposition}
  \label{pro:inf}
  A homogeneous rational potential $V\in\C(q_1,q_2)$ possesses an
  infinite number of Darboux points if and only if $V$ is radial,
  i.e., it has the following form
  \begin{equation}
    \label{eq:r}
    V(\vq):=c(q_1^2+q_2^2)^l, \qquad k=2l, 
  \end{equation}
  for a certain $l\in\Z$.
\end{proposition}
\begin{proof}
  If potential has the form \eqref{eq:r}, then
  \begin{equation*}
    V'(\vq)= \frac{k}{r^2}V(\vq)\vq, \qquad r^2:=q_1^2+q_2^2,
  \end{equation*}
  so for any $\vq\in\C^2$, we have
  \begin{equation*}
    \vq\wedge V'(\vq)=0.
  \end{equation*}
  Hence, in fact, an arbitrary non-zero $\vq\in\C^2$ gives a Darboux
  point.

  On the other hand, if $V$ has an infinite number Darboux points,
  then rational function $f(z)$ given by~\eqref{eq:fz} has infinite
  number of zeros, so it vanishes identically. Thus, we have
  \begin{equation}
    \label{eq:radz}
    (1+z^2)v'(z)-kzv(z)=0.
  \end{equation}
 Using separation of variables we obtain the following solution of this
equation 
  \begin{equation*}
    v(z)=c(1+z^2)^{k/2}.
  \end{equation*}
  Because, by assumptions, $v(z)$ is rational, $k=2l$ for a certain
  integer $l$. The homogenisation of $v(z)$ gives \eqref{eq:r}, and
  this finishes the proof.
\end{proof}
We show that if the considered potential is not radial, then we can
always choose its representative such  that the polynomial $g(z)$ has degree
$r+s$.
\begin{proposition}
  \label{prop:reprezent}
  Assume that the potential has a finite number of Darboux
  points. Then it has a representative such that the corresponding
  polynomial $g(z)$ has degree $r+s$.
\end{proposition}
\begin{proof}
  From formulae~\eqref{eq:wtuvw} it follows that
  \begin{equation}
    \label{eq:u0u1}
    u_0=\widetilde u(0),\qquad u_1={\widetilde u\, }'(0), \mtext{and} 
    w_0=\widetilde w(0),\qquad w_1={\widetilde w\,}'(0)
  \end{equation}
  So, the condition that $g(z)$ has degree $r+s$ can be written in the
  following form

  \begin{equation}
    \label{eq:sz}
    t:= \widetilde w(0){\widetilde u\, }'(0) - {\widetilde
      w\,}'(0)\widetilde u(0)\neq 0.
  \end{equation}
  Let us take a potential $V_A(\vq):=V(A\vq)$ equivalent to
  $V(\vq)$. We assume that $A=\left[\begin{smallmatrix}a & b\\-b & a
  \end{smallmatrix}\right]
  $, with $a^2+b^2\neq 0$.

  Let us assume the $V$ is not radial and for an arbitrary
  $A\in\mathrm{PO}(2,\C)$, the polynomial $g_A(z)$ corresponding to
  $V_A$ has degree smaller that $r+s$. Since $v_A(z)=w_A(z)/u_A(z)$,
  it implies that
  \begin{equation}
    \label{eq:sza}
    t_A:= \widetilde w_A(0){{\widetilde u}\,}'_A(0) - {\widetilde
      w_A}'(0)\widetilde u_A(0)= 0, 
  \end{equation}
  for all $A\in\mathrm{PO}(2,\C)$. We have
  \begin{equation}
    \label{eq:uwa}
    \widetilde u_{A}(\zeta)=\left( a -b \zeta \right)^s
    u(\tau_{A}^{-1}(\zeta)), \qquad 
    \widetilde w_{A}(\zeta)=\left( a -b \zeta \right)^r w(\tau_{A}^{-1}(\zeta)),
  \end{equation}
  and simple calculations give
  \begin{equation}
    \label{eq:tuap}
    \widetilde u_A(0)= a^s u(x), \qquad \widetilde u_A'(0)= a^s \left[
      (1+x^2)u'(x) -sx u(x)  \right],
  \end{equation}
  and
  \begin{equation}
    \label{eq:twap}
    \widetilde w_A(0)= a^r w(x), \qquad \widetilde w_A'(0)= a^r \left[
      (1+x^2)w'(x) -rx w(x)  \right],
  \end{equation}
  where $x:=a/b\in\C$. Now, condition~\eqref{eq:sza} reads
  \begin{equation}
    \label{eq:szax}
    (1+x^2)(w'(x)u(x)-w(x)u'(x)) -(r-s)x w(x)u(x)=0 \mtext{for} x\in\C. 
  \end{equation}
  Dividing by $u(x)^2$ we obtain differential equation
  \begin{equation*}
    (1+x^2)v'(x) -k v(x)=0,  \qquad v(x)=\frac{w(x)}{u(x)}, \quad k=r-s, 
  \end{equation*}
  coinciding with equation~\eqref{eq:radz}. But from the proof of
  Proposition~\ref{pro:inf} we know that for odd $k$ this equation
  does not have a non-zero rational solutions, and for even $k$, the
  only rational solutions are radial. In any case we have a
  contradiction and this finishes the proof.
\end{proof}
The number of different Darboux points of the potential, if it is
finite, can be smaller, then maximal predicted by
Proposition~\ref{pro:maxd}.  It can happens in the following cases.
\begin{enumerate}
\item Potential $V$ has an improper Darboux point;
\item Polynomial $U$ has a multiple linear factor;
\item Polynomial $U(q_1,q_2)$ has a linear factor $(\pm\rmi q_1+q_2)$;
\item Potential $V$ has a multiple proper Darboux point.
\end{enumerate}
It appears that cases 1--3 occurs only if polynomials $g$ and $h$ have
a common root.
Here it is also worth to remark that an  improper Darboux point appears in a
non-generic
situation. In fact, if $z_{\star}$ is an improper Darboux point, then
it is a root of polynomial $g$ and $h$, so they are not relatively
prime.
\begin{proposition}
  \label{pro:ghcz}
  Assume that $k\neq 0$. Then polynomials $g(z)$ and $h(z)$ possess a
  common root $z_{\star}$ if and only if $z_{\star}$ is a multiple
  root of either $w(z)$, or $u(z)$.
\end{proposition}
\begin{proof}
  If $z_{\star}$ is a multiple root of $w(z)$, then
  $w(z_{\star})=w'(z_{\star})=0$ and formulae \eqref{funGH} give
  immediately $g(z_{\star})=h(z_{\star})= 0$.  Similarly, a multiple
  root of $u(z)$ means $u(z_{\star})=u'(z_{\star})=0$ and gives
  immediately $g(z_{\star})=h(z_{\star})= 0$.

  Now we assume that $z_{\star}$ is a common root of $g(z)$ and
  $h(z)$, i.e. $g(z_{\star})=h(z_{\star})= 0$. From formulae
  \eqref{funGH} it follows that
  \[
  zh(z)+g(z)=w'(z)u(z)-u'(z)w(z).
  \]

  Thus for $z=z_{\star}$ we obtain
  \begin{equation}
    w'(z_{\star})u(z_{\star})-u'(z_{\star})w(z_{\star})=0.
    \label{eq:naimprop}
  \end{equation}
  But then from \eqref{funGH}
  $0=g(z_{\star})=-kz_{\star}w(z_{\star})u(z_{\star})$ and
  $0=h(z_{\star})=kw(z_{\star})u(z_{\star})$. For $k\neq0$ we obtain
  that $w(z_{\star})=0$ or $u(z_{\star})=0$. If we choose
  $w(z_{\star})=0$, then \eqref{eq:naimprop} simplifies to
  $w'(z_{\star})u(z_{\star})=0$, that gives $w'(z_{\star})=0$ because
  $u$ and $w$ are relatively prime. It means that $z_{\star}$ is a
  multiple root of $w(z)$. If we choose $u(z_{\star})=0$, then
  \eqref{eq:naimprop} simplifies to $u'(z_{\star})w(z_{\star})=0$, that
  gives $u'(z_{\star})=0$ because $u$ and $w$ are relatively prime. As
  result $z_{\star}$ is a multiple root of $u(z)$.
\end{proof}
In our further considerations we need to know multiplicities of common
roots of $g$ and $u$, as well as, $g$ and $w$. The answer to this
question is given in the following two lemmas.
\begin{lemma}
  \label{lem:wzgz}
  Assume that $\mult(w,z_{\star})=l\geq 1$. Then
  \begin{equation}
    \label{eq:gmulu}
    \mult(g,z_{\star})=
    \begin{cases}
      > l & \mtext{if} z_{\star}\in  \{-\rmi, \rmi\}, \mtext{and} k=2l, \\
      = l & \mtext{if} z_{\star}\in  \{-\rmi, \rmi\}, \mtext{and} k\neq2l, \\
      = l-1 & \mtext{otherwise.}
    \end{cases}
  \end{equation}
\end{lemma}
\begin{proof}
  We can write
  \begin{equation*}
    w=(z-z_{\star})^l\widetilde w  \qquad l, \in \N,
  \end{equation*}
  where $\widetilde w$ is a polynomial such that $\widetilde w(z_{\star})\neq
  0$. Inserting the above expression into the definition of $g$  we
  obtain 
  \begin{equation}
    \label{eq:gum}
    g(z)=(z-z_{\star})^{l-1}\widetilde g(z),
  \end{equation}
  where
  \begin{equation}
    \label{eq:wtg}
    \widetilde
    g=\left[l(z^2+1)-kz(z-z_{\star})\right]\widetilde w
    u+(z^2+1)(z-z_{\star})(\widetilde w' u-u'\widetilde w).
  \end{equation}
  Hence
  \begin{equation}
    \label{eq:gzs}
    \widetilde g(z_{\star})=l(z_{\star}^2+1)\widetilde w(z_{\star}) u(z_{\star}).
  \end{equation}
  Now, if $z_{\star}\not \in \{-\rmi,\rmi\}$, then $\widetilde
  g(z_{\star})\neq0$ because $l\neq 0$, $\widetilde w(z_{\star}) u(z_{\star})\neq
  0$, and $z_{\star}^2+1\neq 0$. Hence, in this case
  $\mult(g,z_{\star})=l-1$.

  Now, assume that $z_{\star}=\rmi$. Then we obtain
  \begin{equation}
    \label{eq:gum1}
    g(z)=(z-z_{\star})^{l-1}\widetilde g(z),
  \end{equation}
  where
  \begin{equation}
    \label{eq:wtg1}
    \widetilde
    g=\left[l(z+\rmi)-kz\right]\widetilde w
    u+(z^2+1)(\widetilde w' u-u'\widetilde w), 
  \end{equation}
  and thus
  \begin{equation*}
    g(\rmi)=\rmi l(2l-k)\widetilde w(\rmi)u(\rmi).
  \end{equation*}
  From the above formula, it follows directly that if $k\neq 2l$, then
  $\mult(g,\rmi)=l$, and otherwise $\mult(g,\rmi)>l$. Similar
  calculations for $z_{\star}=-\rmi$ give that if $k\neq 2l$, then
  $\mult(g,-\rmi)=l$, and otherwise $\mult(g,-\rmi)>l$.  In this way
  we conclude our proof.
\end{proof}
Notice that polynomials $u$ and $w$ enter symmetrically in the
definition~\eqref{funGH} of $g$. Hence we have also the following statement.
\begin{lemma}
  \label{lem:uzgz}
  Assume that $\mult(u,z_{\star})=l\geq 1$. Then
  \begin{equation}
    \label{eq:gmulu1}
    \mult(g,z_{\star})=
    \begin{cases}
      > l & \mtext{if} z_{\star}\in  \{-\rmi, \rmi\}, \mtext{and} k=2l, \\
      = l & \mtext{if} z_{\star}\in  \{-\rmi, \rmi\}, \mtext{and} k\neq2l, \\
      = l-1 & \mtext{otherwise.}
    \end{cases}
  \end{equation}
\end{lemma}
In our considerations we have to know multiplicities of common roots of
polynomials $g(z)$ and $h(z)$.
\begin{lemma}
  \label{lem:crhz}
  Assume that $z_{\star}$ is a common root of polynomials $g(z)$ and
  $h(z)$, and
  \begin{equation*}
    l:=\max \left\{ \mult(w,z_{\star}),  \mult(u,z_{\star}) \right\}.
  \end{equation*}
  Then
  \begin{enumerate}
  \item if $z_{\star}=0$, then
    $\mult(h,z_{\star})>\mult(g,z_{\star})=l-1$;
  \item if $z_{\star}\not\in\{-\rmi,0,\rmi\}$, then
    $\mult(h,z_{\star})=\mult(g,z_{\star})=l-1$;
  \item if $z_{\star}\in\{-\rmi,\rmi\}$, and $l\neq k/2$, then
    $\mult(g,z_{\star})=\mult(h,z_{\star})+1=l$;
  \item if $z_{\star}\in\{-\rmi,\rmi\}$, and $l= k/2$, then
    $\mult(g,z_{\star})>\mult(h,z_{\star})+1=l$.
  \end{enumerate}
\end{lemma}
\begin{proof}
  By Proposition~\ref{pro:ghcz}, $z_{\star}$ is either a multiple root
  of $u(z)$, or a multiple root of $w(z)$. Let us assume that it is a
  multiple root of $w(z)$.  We can write
  \begin{equation*}
    w(z)=(z-z_{\star})^l\widetilde w(z), \qquad \widetilde
    w(z_{\star})\neq 0,
  \end{equation*}
  where $ \widetilde w(z)$ is a polynomial. Then, from definition of
  $h(z)$, see~\eqref{funGH}, we obtain
  \begin{equation*}
    h(z)=(z-z_{\star})^{l-1}\widetilde h(z), 
  \end{equation*}
  where
  \begin{equation*}
    \widetilde h(z)= k(z-z_{\star}) u(z)\widetilde w(z)-z \left[ l u(z)
      \widetilde w(z) + (z-z_{\star}) \left(u'(z)\widetilde w(z) + u(z){\widetilde w}'(z)\right)\right]. 
  \end{equation*}
  Hence
  \begin{equation*}
    \widetilde h(z_{\star})= -lz_{\star}u(z_{\star})  \widetilde
    w(z_{\star}).  
  \end{equation*}
  As $l>0$, $u(z_{\star}) \neq 0$, and $\widetilde w(z_{\star})\neq
  0$, we have that for $z_{\star}\neq 0$, $\mult(h,z_{\star})=l-1$,
  and $\mult(h,0)>l-1$.  Now, combining this with Lemma~\ref{lem:wzgz}
  we obtain the desired results.  In the case when $z_{\star}$ is a
  multiple root of polynomial $u(z)$ the proof is similar.
\end{proof}

Now we want to distinguish those potentials which do not admit any
proper Darboux point.
\begin{lemma}
  \label{lem:bezde}
  Let $V=W/U$ with relatively prime polynomials $W,U\in\C[\vq]$ of
  respective degrees $r$ and $s$ be a homogeneous rational potential of
  degree $k=r-s\neq 0$. If $V$ does not have any proper Darboux point,
  then it is equivalent to the potential
  \begin{equation}
    \label{eq:bezp}
    V(\vq)=c(q_1+\rmi q_2)^{\alpha}(q_1-\rmi q_2)^{\beta}, \quad
    \alpha+\beta=k, \quad \alpha,\beta\in \Z, \qquad c\in\C^{\star}.
  \end{equation}
  except for the case when $k=2l>0$, and either $W$, or $U$, has a
  factor $(q_1\pm\rmi q_2)$ with multiplicity~$l$.
\end{lemma}
\begin{proof}
  From definitions~\eqref{zb2} it follows that if $V$ has no proper
  Darboux points, then all roots of polynomial $g$ are either roots of
  polynomial $u$, or polynomial $h$. If a root of $g$ is also a root
  of $h$, then by Proposition~\ref{pro:ghcz}, this common root is
  either a root of polynomial $u$, or a root of polynomial $w$. In
  effect, all roots of $g$ are either roots $u$, or roots of $w$.

  We notice, that the potential $V$ has only a finite number of
  Darboux points. In fact, for otherwise it is radial and has a proper
  Darboux point. Thanks to this fact, by
  Proposition~\ref{prop:reprezent}, we can assume that degree of
  polynomial $g$ is $r+s$.  Polynomials $u$ and $w$ have the
  form~\eqref{eq:wuab} with all $\alpha_i>1$ and all $\beta_i>1$. We
  can write polynomial $g$ in the following form
  \begin{equation}
    \label{eq:gnd}
    g(z)=\delta \prod_{i=1}^{l'}(z-a_i)^{\alpha'_i}\prod_{i=1}^{m'}(z-b_i)^{\beta'_{i}}
  \end{equation}
  with $l'\leq l$, $m'\leq m$.  From assumptions it follows that for
  $k=2l>0$, if $a_j\in\{-\rmi,\rmi\}$ for a certain $1\leq j\leq l'$,
  then $\alpha_j\neq l$. This is why, by Lemma~\ref{lem:wzgz},
  $\alpha'_j\leq\alpha_j$, for $1\leq j\leq l'$. By the same reason
  $\beta_j'\leq\beta_j$ for $1\leq j\leq m'$.

  Now, let us assume that $l'+m' >2$. Then, either there exists $1\leq
  j\leq l'$ such that $\alpha_j'=\alpha_j-1$, or there exists $1\leq
  j\leq m'$ such that $\beta_j'=\beta_j-1$. By
  Proposition~\ref{prop:reprezent}, we can assume that $\deg g = r+s$,
  thus from~\eqref{eq:gnd} we obtain
  \begin{equation}
    \label{eq:rs}
    r+s = \sum_{i=1}^{l'}\alpha_i' + \sum_{i=1}^{m'} \beta'_{i}\leq \sum_{i=1}^{l'}\alpha_i + \sum_{i=1}^{m'} \beta_{i}-1<\sum_{i=1}^{l}\alpha_i + \sum_{i=1}^{m} \beta_{i}=r+s.
  \end{equation}
  But this is impossible. This is why $l'+m' \leq 2$.

  If $l'+m'<l+m$, then
  \begin{equation*}
    r+s = \sum_{i=1}^{l'}\alpha_i' + \sum_{i=1}^{m'} \beta'_{i}<\sum_{i=1}^{l}\alpha_i + \sum_{i=1}^{m} \beta_{i}=r+s.
  \end{equation*}
  But it is impossible. So, in effect, $l'+m'=l+m$ and it is possible
  only when $l'=l$ and $m'=m$.

  If for a certain $1\leq i \leq l$, $a_i\not\in\{-\rmi, \rmi\}$, then
  $\alpha_i'=\alpha_i-1$, and
  \begin{equation*}
    r+s = \sum_{i=1}^{l}\alpha_i' + \sum_{i=1}^{m} \beta'_{i}<\sum_{i=1}^{l}\alpha_i + \sum_{i=1}^{m} \beta_{i}-1<r+s.
  \end{equation*}
  But it is impossible, so $a_i\in\{-\rmi,\rmi\}$, for $1\leq i, \leq
  l$. In a similar way we show that $b_i\in\{-\rmi,\rmi\}$, for $1\leq
  i, \leq m$. In effect the potential has the form~\eqref{eq:bezp},
  and this finishes our proof.
\end{proof}

Now we characterize potentials with multiple Darboux points.  We say
that $z_{\ast}$ is a multiple proper Darboux point if
$g(z_{\ast})=g'(z_{\ast})=0$, $h(z_{\ast})\neq 0$ and $u(z_{\ast})\neq
0$.
\begin{proposition}
  \label{propmult}
  If a homogeneous rational potential of degree homogeneity
  $k\in\Z^{\ast}$ with two degrees of freedom has a multiple proper
  non-isotropic Darboux point, then it is equivalent to a potential of
  the form \eqref{pot2zm} with coefficients satisfying the following
  conditions
  \begin{enumerate}
  \item $w_{r-1}u_s=u_{s-1}w_r$,
  \item $kw_ru_s=2\left(w_{r-2}u_s-w_ru_{s-2} \right)$,
  \item $kw_ru_s\neq 0$.
  \end{enumerate}
\end{proposition}

\begin{proof}
  Let us calculate $g'(z)$
  \begin{equation}
    g'(z)=(1+z^2)(w''(z)u(z)-u''(z)w(z))+(2-k)z(w'(z)u(z)-u'(z)w(z))-kw(z)u(z).
    \label{eq:difg}
  \end{equation}
  From above-mentioned considerations we know that at a multiple
  proper Darboux point we have $g(z_{\star})=g'(z_{\star})=0$ and
  $h(z_{\star})\neq 0$.  Because, by assumption, the considered Darboux
  point is not isotropic we can move it using a rotation into a point
  with affine coordinate $z_{\star}=0$.  Using explicit forms of
  polynomial $w(z)$ and $u(z)$ substituted into conditions
  $g(0)=g'(0)$, $h(0)\neq 0$ and $u(0)\neq0$ we obtain the thesis of
  this proposition.
\end{proof}

We recall that $z_{\star}$ is a multiple improper Darboux point iff
$g(z_{\star})=g'(z_{\star})=h(z_{\star})=0$ and $u(z_{\star})\neq0$.
\begin{proposition}
  \label{midp}
  If a homogeneous rational potential of degree homogeneity
  $k\in\Z^{\star}$ with two degrees of freedom has a multiple improper
  non-isotropic Darboux point, then it is equivalent to a potential of
  the form \eqref{pot2zm} with coefficients satisfying the following
  conditions
  \[
  w_r=w_{r-1}=w_{r-2}=0,\quad \text{and}\quad u_s\neq 0.
  \]
\end{proposition}

\begin{proof}
  As in the previous proposition we can assume without loss of
  generality that $z_{\star}=0$. Then conditions $g(0)=g'(0)=h(0)$ give
  the following system
  \[
  w_{r-1}u_s=u_{s-1}w_r,\qquad kw_ru_s=2\left(w_{r-2}u_s
    -w_ru_{s-2}\right),\qquad kw_ru_s= 0,
  \]
  and taking into account that, by assumption $k\neq 0$ and
  $u(0)=u_s\neq 0$ the solution of this system gives the thesis of our
  proposition.
\end{proof}
An interested reader can formulate analogous propositions for potentials with
multiple isotropic Darboux points. Without loss and generality
one can assume that the affine coordinate of an isotropic Darboux point is
$z_{\star}=\rmi$.

\subsection{Relations between non-trivial eigenvalues at various
  proper Darboux points}
\label{NFI}
We start this section with calculation of non-trivial eigenvalue of
$V''(\vd)$ in terms of affine coordinate $z$ of Darboux point $[\vd]$.
We assume that the considered Darboux point is proper and it
satisfies $V'(\vd)=\gamma\vd$ with $\gamma\in\C^{\star}$.  Hence, our
aim is to evaluate function
\begin{equation}
  \label{eq:lq}
  \lambda=\gamma^{-1}\tr V''(\vq)-(k-1)
\end{equation}
at a point $\vq\neq\vzero$ satisfying $V'(\vq)=\gamma\vq$, and express
the result by means of the affine coordinate $z=q_2/q_1$.

Since $V(q_1, q_2)=q_1^kv(z)$ we can express partial derivatives of
$V$ at point $(q_1,q_2)$ in terms of variables $(q_1,z)$. They have
the following forms
\begin{equation}
  \begin{split}
    &\dfrac{\partial V}{\partial
      q_1}=q_1^{k-1}\left(kv-zv'\right),\qquad
    \dfrac{\partial V}{\partial q_2}=q_1^{k-1}v',\\
    &\dfrac{\partial^2 V}{\partial q_1^2}=q_1^{k-2}\left\{(k-1)\left(
        kv-zv' \right)+z\left(zv''-(k-1)v'\right)\right\},\qquad
    \dfrac{\partial^2 V}{\partial q_2^2}=q_1^{k-2}v''.
  \end{split}
\end{equation} 
As $v=w/u $, the first component of equality $V'(\vq)=\gamma \vq$
gives
\[
\gamma=q_1^{k-2}\left(kv-zv'\right)=\dfrac{kwu-z(w'u-wu')}{u^2}=\dfrac{h}{u^2},
\]
so we have immediately that $q_1^{k-2}=\gamma u^2/h$. Then
\[
\begin{split}
  &\lambda=\gamma^{-1}\tr V''(\vq)-(k-1)=\gamma^{-1}\left\{
    q_1^{k-2}[(k-1)(kv-zv')]+q_1^{k-2}[(1+z^2)v''-(k-1)zv']
  \right\}\\
  &-(k-1)=\dfrac{u^2}{h}[(1+z^2)v''-(k-1)zv'].
\end{split}
\]
Taking into account that
\[
v'=\dfrac{w'u-wu'}{u^2},\quad
v''=\dfrac{w''u-wu''}{u^2}+\dfrac{2u'}{u^3}(wu'-w'u)
\]
we obtain
\[
\begin{split}
  &\lambda=\dfrac{u^2}{h}\left\{
    \dfrac{(1+z^2)(w''u-wu'')-(k-1)z(w'u-wu')}{u^2}-\dfrac{2u'(1+z^2)(w'u-wu')}{u^3}
  \right\}.
\end{split}
\]
But the above expression is evaluated only at proper Darboux points,
so $g=0$ at these points. Thus, we have $(1+z^2)(w'u-wu')=kzwu$, and
the final form is the following
\begin{equation}
  \lambda=\dfrac{1}{h}\left[(1+z^2)(w''u-wu'')-(k-1)zw'u-(k+1)zu'w\right]=
  \dfrac{g'}{h}+1.
\end{equation} 
The last equality one can verify by the direct check. We see that if a
proper Darboux point is multiple, then $\lambda=1$.

If the considered Darboux point $z_{\star}$ is simple, then the
inverse of the corresponding shifted eigenvalue
$\Lambda_{\star}=\lambda_{\star}-1$
is given by the following formula
\begin{equation}
  \label{eq:rat}
  \frac{1}{\Lambda_{\star}}=\frac{h(z_{\star})}{g'(z_{\star})}. 
\end{equation}
It is crucial to notice that in the right hand side of the above formula is the residue of rational
function $h(z)/g(z)$ at point $z=z_{\star}$.

\subsubsection{Proof of Theorem~\ref{thm:rel}}

We choose such a representative potential that the polynomial $g$
given by \eqref{funGH} has degree $r+s$, see
Proposition~\ref{prop:reprezent} and comments after it.  This
guarantees that all Darboux points are located in the affine part of
the projective line $\CP^1$, and are roots of $g$.  Next, we define a
meromorphic differential form $\omega$ which in the affine part of
$\CP^1$ is given by
\[
\omega(z)=\dfrac{h(z)}{g(z)}\operatorname{d}z,
\]
and, in a neighbourhood of infinity, in local coordinate $\zeta=1/z$
it has the form
\begin{equation}
  \label{eq:ot}
  \widetilde
  \omega(\zeta)=\omega\left(\dfrac{1}{\zeta}\right)=-\dfrac{h(1/\zeta)}{g(1/\zeta)
  } \dfrac{ \operatorname { d } \zeta } { \zeta^2}.
\end{equation}
At first we assume that polynomials $g$ and $h$ are relatively prime
and all roots of $g$ are simple. This implies that $l=r+s$ simple
roots of $g$ are proper Darboux points.  Taking into account the above
facts we can calculate residues $\res(\omega,z_i)$ for
$i=1,\ldots,r+s$
\begin{equation}
  \label{lambdaGH}
  \res(\omega,z_i)=\dfrac{h(z_i)}{g'(z_i)}=\dfrac{1}{\Lambda(z_i)}=
  \dfrac{1}{\Lambda_i}.
\end{equation}
Now, we calculate residue at infinity
$\res(\omega,\infty)=\res(\widetilde\omega,0)$.  Let us write
polynomials $g$ and $h$ in the following form
\begin{equation}
  \label{eq:}
  g=\sum_{i=0}^{r+s}g_{r+s-i}z^i, \qquad h=\sum_{i=0}^{r+s-1}h_{r+s-1-i}z^i,.
\end{equation}
Using definition \eqref{funGH} we find that
\begin{equation}
  \label{eq:g0h0}
  g_0=-h_0= u_1w_0-u_0w_1\neq 0, 
\end{equation}
see Proposition~\ref{pro:maxd}.

Now, we rewrite \eqref{eq:ot} in the following form
\begin{equation}
  \label{eq:ott}
  \widetilde\omega(\zeta)=  \frac{A(\zeta)}{\zeta} \rmd \zeta,
  \mtext{where}A(\zeta):= -\dfrac{h(1/\zeta)}{\zeta g(1/\zeta)
  } 
\end{equation}
We show that $A(0)=1$. In fact, we can write
\begin{equation*}
  A(\zeta)=\frac{a(\zeta)}{b(\zeta)}, \qquad
  a(\zeta):=-\zeta^{r+s-1}h(1/\zeta), \quad b(\zeta):=\zeta^{r+s}g(1/\zeta)
\end{equation*}
Clearly, $a(\zeta)$ and $b(\zeta)$ are polynomials and
$a(0)=b(0)=g_0$. Hence, $A(0)=1$ as we claimed.
\[
\res(\omega,\infty)=\res(\widetilde\omega,0)= 1 .
\]
According to the global residue theorem we have
\[
\sum_{i=1}^{r+s}\dfrac{1}{\Lambda_i}=\sum_{i=1}^{r+s}\res(\omega,
z_i)=-\res(\omega,\infty)=-1.
\]
In this way we proved Theorem~\ref{thm:rel} for $l=r+s$.

Now, we consider the general case.  Let $z_1, \ldots, z_l$ denotes
roots of $g$ corresponding to proper Darboux points.  By assumption
C2, if $g(z_{\star})=0=u(z_{\star})$, then $z_{\star}\not \in\{-\rmi,\rmi\}$.
Them, since $\mult(g,z_{\star})\geq1$ by Lemma~\ref{lem:uzgz} and assumption
C2, $\mult(u,z_{\star})=1+\mult(g,z_{\star})>1$. So, $z_{\star}$ is a multiple
root of $u(z)$. Thus by
Proposition~\ref{pro:ghcz}, $z_{\star}$ is  a root of $h(z)$.  In
effect, all roots of $g(z)$ which are not proper Darboux points are
also roots of $h(z)$.  By assumptions $C2$, $C3$, and
Lemma~\ref{lem:crhz} we know that if $z_{\star}$ is a common root of
$g$ and $h$, then $\mult(g,z_{\star}) \leq \mult(h,z_{\star})
$. Hence, we have $g=f\bar{g}$, and $h=f\bar{h}$, where $f$, $\bar{g}$
and $\bar{h} $ are polynomials, and $\bar{g}(z_{\star}) \neq 0$.  We
denote also $\Lambda_i=\Lambda(z_i)$ for $i=1,\ldots,l$.  Now, the
form $\omega$ reads
\[
\omega(z)=\dfrac{h(z)}{g(z)}\operatorname{d}z=
\dfrac{\bar{h}(z)}{\bar{g}(z)}\operatorname{d}z.
\]
It has poles at $z_1,\ldots, z_l$ and at the infinity. Notice that we
have
\[
\dfrac{1}{\Lambda_i}=
\dfrac{h(z_i)}{g'(z_i)}=\dfrac{f(z_i)\bar{h}(z_i)}{f'(z_i)\bar{g}(z_i)+f(z_i)\bar{g}'(z_i)}=\dfrac{\bar{h}(z_i)}{\bar{g}'(z_i)}=\res(\omega,z_i),\quad
i=1,\ldots,l,
\]
where $f(z_i)\neq 0$ and $\bar{g}(z_i)=0$ for $i=1,\ldots, l$.
Additionally, we have also $\res(\omega,\infty)=1$, hence by the
global residue theorem we obtain~\eqref{eq:relka1} and this finishes
the proof.

\subsubsection{Proof of Theorem~\ref{thm:rell2} }
If $s_{+}=s_{-}=0$, and $r_{+},r_{-}\in\{0,1\}$, then assumptions of
Theorem~\ref{thm:rel} are satisfied.  So, in this case our theorem is
valid.

Now, we consider general case. We can write
\begin{equation}
  \label{eq:wupm}
  w(z)=(z-\rmi)^{r_{+}}(z+\rmi)^{r_{-}}\widetilde w(z), \qquad
  u(z)=(z-\rmi)^{s_{+}}(z+\rmi)^{s_{-}}\widetilde u(z), 
\end{equation}
where $\widetilde w(z)$ and $\widetilde u(z)$ are polynomials, and
$\widetilde w(\pm\rmi)\neq 0 $, $\widetilde u(\pm\rmi)\neq 0 $. As we
assumed that $u(z)$ and $w(z)$ are relatively prime, at most one number
in pairs $(r_{+},s_{+})$ and $(r_{-},s_{-})$ is different from zero.

Calculating function $g(z)$ and $h(z)$ using \eqref{eq:wupm} we obtain
\begin{equation}
  \label{bar}
  \begin{split}
    g(z)=&(z-\operatorname{i})^{r_++s_{+}}(z+\operatorname{i})^{r_-+s_{-}}
    \bar{g}(z),\\
    h(z)=&(z-\operatorname{i})^{r_++s_{+}-1}(z+\operatorname{i})^{r_-+s_{-}-1}
    \bar{h}(z),
  \end{split}
\end{equation}
where
\begin{equation}
  \begin{split}
    \bar{g}(z)=&(1+z^2)(\widetilde{w}'(z)\widetilde{u}(z)-\widetilde{u}
    '(z)\widetilde {w}(z))\\ &+\left[
      (r_+-s_+)(z+\rmi)+(r_--s_-)(z-\rmi)-kz
    \right]\widetilde{w}(z)\widetilde{u}(z),\\
    \bar{h}(z)=&\left\{
      k(z^2+1)-z\left[(r_+-s_+)(z+\rmi)+(r_--s_-)(z-\rmi)\right]\right\}\widetilde
    {w}(z)\widetilde{u}(z)\\
    & - z(z^2+1)(\widetilde{w}'(z)\widetilde{u}(z)-\widetilde{u}
    '(z)\widetilde {w}(z)).
  \end{split}
  \label{eq:ghbar}
\end{equation}
Thanks to above equations we obtain differential form $\omega(z)$ as
follows
\[
\omega(z)=\dfrac{h(z)}{g(z)}\operatorname{d}z=\dfrac{\bar{h}(z)}{g_1(z)}
\operatorname{d}z,\quad \text{where} \quad g_1(z)=(1+z^2)\bar{g}(z)
\]
and polynomials $\bar{h}(z)$, $\bar{g}(z)$ are relatively prime.

Notice that all finite poles of form $\omega(z)$ are all proper
Darboux points $z_i$, infinity and $\pm\rmi$. Under assumptions of
Theorem~\ref{thm:rell2}, all poles are simple.  Calculations of
residues at proper Darboux points as well as at infinity are similar
to these in the proof of the previous theorem and we obtain
\[
\operatorname{res}(\omega,z_i)=\dfrac{1}{\Lambda_i},\qquad
\operatorname{res}(\omega,\infty)=1.
\]

Residues at $z=\pm\rmi$ are calculated in the following way
\[
\operatorname{res}(\omega,\pm\rmi)=\dfrac{\bar{h}(\pm\rmi)}{g_1'(\pm\rmi)}.
\]
Because
\[
g_1'(z)=2z\bar{g}(z)+(z^2+1)\bar{g}'(z),
\]
thus we obtain
\[
g_1'(\pm\rmi)=\pm2\rmi \bar{g}(\pm\rmi)
\]
but
\[
\begin{split}
  &\bar
  g(\rmi)=2\rmi\left(r_+-s_+-\dfrac{k}{2}\right)u(\rmi)\widetilde{w}(\rmi),\qquad
  \qquad
  \bar{h}(\rmi)=2(r_{+}-s_{+})u(\rmi)\widetilde{w}(\rmi),\\
  &\bar
  g(-\rmi)=-2\rmi\left(r_--s_--\dfrac{k}{2}\right)u(-\rmi)\widetilde{w}(-\rmi),
  \quad \bar{h}(-\rmi)=-2 (r_{-}-s_-)u(-\rmi)\widetilde{w}(-\rmi).
\end{split}
\]
Now we finish calculation of the residues
\[
\operatorname{res}(\omega,\pm\rmi)=\dfrac{\bar{h}(\pm\rmi)}{\pm2\rmi
  \bar g(\pm\rmi)}=\dfrac{r_{\pm}-s_{\pm}}{k-2(r_{\pm}-s_{\pm})}.
\]
Applying the global residue theorem we obtain
relation~\eqref{eq:relka2}.

\subsection{Finiteness of admissible spectra for integrable
  potentials}
\label{sec:finiteness}

At first we recall the notion of a limit point of a set.  Let $T$ be a
topological space and $A$ a subset of $T$. A point $x\in T$ is a limit
point of $A$, iff in every open neighbourhood $U$ contains points of
$A$, i.e., $U\cap A\neq\emptyset$.

We start this section with a certain technical lemma.
 Let $\scX$ be a discrete subset of $\R_+=(0,\infty)$ such that its
  only limit point is 0. We consider the following equation
  \begin{equation}
    X_1+\cdots+X_m=c,\qquad c>0,
    \label{eq:eqX}
  \end{equation} 
  and we look for its solutions $\vX=(X_1,\ldots,
  X_m)\in\underbrace{\scX \times \cdots \times \scX}_{m}=\scX^m$.  We
  prove the following lemma.
  \begin{lemma}
    \label{lem:dum}
    For arbitrary $c>0$ equation~\eqref{eq:eqX} has at most a finite
    number of solutions in $\scX^m$.
  \end{lemma}
  \begin{proof}
    We prove this lemma by induction with respect to $m$. For $m=1$
    the statement of the lemma is evidently true. So, assume that it
    is true for $m-1$. We have to show that it is true for $m$.

    First we notice that if $\vX=(X_1,\ldots,X_m)$ is a solution of
    equation \eqref{eq:eqX}, than at least one its component, let us
    say $X_m$, is not smaller than $c/m$.  In fact, if we have
    $X_i<c/m$, for $i=1,\ldots,m$, then
    \[
    c=X_1+\ldots+X_m<m\dfrac{c}{m}=c.
    \]
    This contradiction proves our claim. Hence without loss of the
    generality we can assume that $X_m\geq c/m$ and we rewrite
    equation \eqref{eq:eqX} in the following form
    \begin{equation}
      \label{x2}
      X_1+\ldots+X_{m-1}=c-X_m.
    \end{equation}
    Now, it is enough to notice that $X_m$ belongs to the set
    $\{X\in\cX\ |\ X>c/m\}$ which is finite because $\scX$ is discrete
    and its only limit point is 0. Thus we have only a finite number
    of possible choices for $X_m$. For each such choice, equation
    \eqref{x2}, considered as equation with unknown
    $(X_1,\ldots,X_{m-1})$ has, by the induction assumption, a finite
    number of solutions. Thus, the number of solution of equation
    \eqref{eq:eqX} is finite.
  \end{proof}

We recall that   $\mathscr{M}_k$ denotes  a subset of rational numbers
$\lambda$
specified by the table in Theorem~\ref{thm:MoRa} for a given $k$,
e.g., for $\abs{k}>5$ we have
\begin{equation}
  \mathscr{M}_k=\defset{ p + \dfrac{k}{2}p(p-1)}{ p\in\Z}\cup
  \defset{\dfrac 1
    {2}\left[\dfrac {k-1} {k}+p(p+1)k\right]}{p\in\Z}.
\end{equation}
We define the following sets
\begin{equation}
  \label{eq:Ik}
  \scI_{k}:=\defset{\Lambda\in\Q}{\lambda=\Lambda+1\in\scM_k},
\end{equation}
and
\begin{equation}
  \label{eq:xk}
  \scX_k:=\defset{X\in\Q}{X=\frac{1}{\Lambda}, \quad\Lambda\in\scI_{k}}.
\end{equation}
In order to describe properties of the above defined sets its
convenient to consider them as subsets of compactified real line
$\R\cup\{\infty\}\simeq\R\PP^1$.
\begin{proposition}
  Assume that $k\in\N\setminus\{2\}$. Then
  \begin{enumerate}
  \item If $\Lambda\in \scI_k$, then $\Lambda\geq-1$.
  \item The infinity $\{\infty\}$ is the only limit point of set
    $\scI_k$ in $\R\cup\{\infty\}$.
  \item If $x\in \scX_k$, then either $x\leq-1$, or $x>0$,
    i.e. $\scX_k\subset \R\setminus (-1,0]$.
  \item Only $\{0\}$ is a limit point of $ \scX_k$.
  \end{enumerate}
\end{proposition}
\begin{proposition}
\label{pro:neg}
  Assume that $k\in-\N\setminus\{-2\}$. Then
  \begin{enumerate}
  \item If $\Lambda\in\scI_k$, then $\Lambda\leq0$.
  \item The minus infinity $\{-\infty\}$ is the only limit point of
    set $\scI_k$ in $\R\cup\{-\infty\}$.
  \item If $x\in \scX_k$, then $x<0$, i.e. $\scX_k\subset
    (-\infty,0]$.
  \item Only $\{0\}$ is the  limit point of $\scX_k$.
  \end{enumerate}
\end{proposition}
Easy proofs of the above propositions we left to the reader. Now we
can pass to the proof of Theorem~\ref{thm:finiteness}.

\begin{proof}
  At first we assume $k\in\N\setminus\{2\}$.  We can write
  equation~\eqref{eq:relka2} in the form
  \begin{equation}
    \label{proofrel}
    \sum_{i=1}^l X_i=c, \qquad c\in\Q,
  \end{equation}
  where $X_i:=1/\Lambda_i$, for $i=1,\ldots,l$.  Let $N$ denote the
  number of solutions of this equation in $\scX_k^l$, and $N_p$ for
  $0\leq p\leq l$, denote the number of solutions of this equation in
  $\scX_k^l$ such that they have $p$ negative components.
We have
\[
N=\sum_{p=0}^lN_p.
\]
We will show that $N_p<\infty$ for $0\leq p\leq l$. In fact, set $\scX_k$ has
only a finite number of negative elements. Thus for a given $p$ we have a finite
number of choices of $p$ negative elements $X_i$. Let
$X_l,X_{l-1},\ldots,X_{l-p+1}\in\scX_k$ be negative. Then
\[
\sum_{i=1}^{l-p}X_i=c'>0,
\]
where $X_1,\ldots,X_{l-p}\in\scX_k$ and $X_i>0$ for $i=1,\ldots,l-p$. By
  Lemma~\ref{lem:dum}, there is only a finite number of solution of
  this equation, and this finishes the proof for
  $k\in\N\setminus\{2\}$. 

For $k\in -\N\setminus\{-2\}$ we write  equation~\eqref{eq:relka2} in the form
  \begin{equation}
    \label{proofrel1}
    -\sum_{i=1}^l X_i=c, \qquad c\in\Q,
  \end{equation}
where $X_i:=1/\Lambda_i$, for $i=1,\ldots,l$. Putting $Y_i:=-X_i$   for $i=1, \ldots, n$, and $\scY_k:=-\scX_k$ we have to show that number of solutions  of equation 
\begin{equation*}
\sum_{i=1}^lY_i =c, \qquad c\in\Q, 
\end{equation*}
such that $(Y_1, \ldots, Y_l)\in\scY_k^l$ is finite.  By
Proposition~\ref{pro:neg}, if $y\in\scY_k$, then $y>0$ and only $0$ is
a limit point of $\scY_k$. Hence we can proceed like in the case of
positive $k$.
\end{proof}
\section{Applications}
\label{appl}

After determination of all necessary integrability conditions for a
given $r$ and $s$:
\begin{itemize}
\item all non-trivial eigenvalues $\Lambda_i$ of the Hessian must
  belong to $\scM_{k}$, where $k=r-s$, and
\item appropriate relations between these eigenvalues calculated at
  all proper Darboux points must exist
\end{itemize}
one can try to classify all integrable rational potentials.

This problem will be considered elsewhere. Here we only show two
one-parameter families of potentials satisfying these conditions.
Both belong to the special class of rational potentials of the
following form
\begin{equation}
  \label{specjal}
  V=q_1^{-s}\sum_{i=0}^r v_{r-i}q_1^{r-i}q_2^i,\quad v_i\in\C.
\end{equation} 
In such a case
\[
w=v=\sum_{i=0}^r v_{r-i}z^i,\qquad u=1,\qquad v=w.
\]
The maximal number of proper Darboux points is $r+1$ provided
$s\neq0$. In this section we show two families of such potentials that
satisfy all known integrability conditions and between them some
integrable cases were found.

\subsection{Example with $k\in\Z_-\setminus \{-2\}$}
\label{k-}
In this section we will consider a family of potentials
\eqref{specjal} that are homogeneous of degree $k=-\kappa$, where
$\kappa>0$. We assume that they possess maximal number $r+1$ proper
Darboux points and the non-trivial eigenvalues are the same at all
proper Darboux points i.e.  $\lambda_1=\ldots=\lambda_{r+1}=\lambda$
and belong to the first item of the Morales-Ramis table. This means
that
\[
\lambda= p-\frac{1}{2}\kappa p(p-1),\qquad \text{for some}\quad p\in\Z
\]
or
\begin{equation}
  \label{c1}
  \Lambda:=\lambda-1=\frac{1}{2}\left(-\kappa p^2+(\kappa+2)p-2\right).
\end{equation}

If we assume that we have $r+1$ simple proper Darboux points different
from $\pm\rmi$ and this number is exactly equal to
\begin{equation}
  r+1=-\frac{1}{2}\left(-\kappa p^2+(\kappa+2)p-2\right),
\end{equation}
then relation 
\begin{equation}
 \sum_{i=1}^{r+1}\frac{1}{\Lambda}=-1
\label{eq:relpotek}
\end{equation} 
 obviously
holds. Now we will reconstruct the form of potential starting from
these values of spectra of Hessians at all proper Darboux points.

Because the number of proper Darboux points is maximal thus function
$g(z) $ defined in~\eqref{funGH} has only simple roots and can be
written as
\begin{equation}
  \label{c2}
  g(z)=\alpha\prod_{i=1}^{r+1}(z-z_i),\qquad \alpha\in\C^{\star}.
\end{equation}
Using above form of $g(z)$ and decomposition on simple fractions we
can calculate $h(z)/g(z)$
\[
\frac{h(z)}{g(z)}=\sum_{i=1}^{r+1}\frac{h(z_i)}{g'(z_i)}\frac{1}{z-z_1}=\sum_{
  i=1
}^{r+1}\frac{1}{\Lambda_i}\frac{1}{z-z_i}=\frac{1}{\Lambda}\sum_{i=1}^{r+1}\frac
{1}{z-z_i}=\frac{1}{\Lambda}\frac{g'(z)}{g(z)}.
\]
Putting equation~\eqref{c1} into above formula we find
\[
h(z)=\frac{2}{-\kappa p^2+(\kappa+2)p-2}g'(z).
\]
Finally if we use expressions~\eqref{funGH} for function $g(z)$ and
$h(z)$ we obtain second order linear differential equation
\[
(1+z^2)\frac{\operatorname{d}^2v(z)}{\operatorname{d}z}+\left(\kappa+2-\frac{\kappa
    p^2-(\kappa+2)p}{2}\right)z\frac{\operatorname{d}v(z)}{\operatorname{d}z}+\left(\kappa-\kappa\frac{\kappa
    p^2-(\kappa+2)p}{2}\right)v(z)=0.
\]
Applying transformation in a form $z\mapsto x=\rmi z$, where
$\rmi^2=-1$ we obtain
\begin{equation*}
  (1-x^2)\frac{\operatorname{d}^2v(x)}{\operatorname{d}x^2}-\left(\kappa+2-\frac{
      \kappa p^2-(\kappa+2)p}{2} 
  \right)x\frac{\operatorname{d}v(x)}{\operatorname{d}x}+\frac{1}{2}\kappa
  p(\kappa p-(\kappa+2))v(x)=0.
\end{equation*}
If we rewrite this equation as
\begin{equation}
  \label{c4}
  (1-x^2)\frac{\operatorname{d}^2v(x)}{\operatorname{d}x^2}
  +\left(\beta-\alpha-(\alpha+\beta+2)x
  \right)\frac{\operatorname{d}v(x)}{\operatorname{d}x}
  +r(r+\alpha+\beta+1)v(x)=0
\end{equation}
where
\begin{equation}
  \alpha=\beta=\frac{2(p-1)+\kappa(2+p-p^2)}{4},\qquad
  r=\frac{p\left(\kappa(p-1)-2 \right)}{2},
  \label{eq:wyry}
\end{equation}
then we recognise immediately that this is equation defining Jacobi
polynomials $P_r^{(\alpha,\beta)}(x)$ of degree $r$. Jacobi
polynomials can be written as
\[
P_r^{(\alpha,\beta)}(x)=2^{-r}\sum_{i=0}^r{r+\alpha \choose i}{
  r+\beta \choose r-i}(x-1)^{r-i}(x+1)^i.
\]
It means that potential $v(z)$ has the form
\[
v(z)=P^{(\alpha,\beta)}_r(\rmi z)=P^{(\alpha,\beta)}_r\left(\rmi
  \frac{q_2}{q_1}\right).
\]
and its homogenization gives the final form of our potential
\begin{equation}
  V(q_1,q_2)=q_1^{-\kappa}v(\rmi z)=q_1^{-\kappa}P_r^{(\alpha,\beta)}\left(\rmi
    \frac{q_2}{q_1}\right).
  \label{eq:dorki}
\end{equation}
Notice that we obtain two-parameter family of potentials which satisfy
all known necessary integrability conditions. The role of parameters
play $p\in\Z$ and $\kappa\in\N$, see expressions in \eqref{eq:wyry}.

\subsection{Special class of integrable rational potentials}
\label{dorizzi}
We can try to fix $p$ and $\kappa$ and use the direct method in order
to find additional first integral. Condition $r>0$, where $r$ is given
by \eqref{eq:wyry} gives the following inequality for $\kappa$
\[
\kappa>\dfrac{2}{p-1}.
\]
For $p=1$ and small $k$ we made such experiments but without
success. The situation changes for $p=2$ that gives $r+1=\kappa-1$
i.e. $r=\kappa-2$ for $\kappa>2$ and $\alpha=\beta=1/2$. In this case
we have
\[
P_r^{\left(\frac{1}{2},\frac{1}{2}\right)}(z)=\dfrac{1}{(r+1)!}\left(\dfrac{3}{2
  } \right)_rU_r(z)
\]
where $(x)_r$ means the Pochhammer symbol, see e.g. formulae for
Jacobi polynomials on Wolfram page \cite{Wolfram}. Here $U_r$ is the
Chebyshev polynomial of the second kind determined by the following
formula
\[
U_r(z)=\sum_{i=0}^{\left[\frac{r}{2}\right]}\dfrac{(-1)^i(r-i)!(2z)^{r-2i}}{
  i!(r-2i)!}.
\]
Substitution $z=\rmi q_2/q_1$ gives the following form
\begin{equation}
  \begin{split}
    &V(q_1,q_2)=q_1^{ -r-2
    }P_r^{\left(\frac{1}{2},\frac{1}{2}\right)}\left(\rmi\dfrac{q_2}{q_1}\right)=
    \dfrac{C}{q_1^{2r+2}}\sum_ { i=0 } ^ { \left [ \frac {r } { 2 }
      \right]}2^{-2i}\dfrac{(r-i)!}{i!(r-2i)!}q_1^{ 2i } q_2^ { r-2i
    }\\
    &=\dfrac{C}{q_1^{2r+2}\rho}\left[\left(\frac{\rho+q_2}{2}\right)^{r+1}
      +(-1)^r\left(\frac{\rho-q_2} { 2 }\right)^{r+1}\right]
  \end{split}
\end{equation}
where $\rho=\sqrt{q_1^2+q_2^2}$ and
\[
C=\frac{\left(\dfrac{3}{2 } \right)_r(r\rmi)^r}{(r+1)!}.
\]

Such potentials have appeared in \cite{Dorizzi:83::}.  They can be
written also as
\begin{equation}
  V_n=\dfrac{1}{\rho}\left[\left(\frac{\rho+q_2}{2}\right)^{n+1}
    +(-1)^n\left(\frac{\rho-q_2}
      { 2 }\right)^{n+1}\right],
  \label{s1}
\end{equation} 
for negative $n=-r-2$.

\begin{example}
  \label{e1}
  The first two non-trivial integrable rational potentials are
  following
  \[
  \begin{split}
  V_{-3}(q_1,q_2)&=q_1^{-3}P_1^{\left(\frac{1}{2},\frac{1}{2}\right)}\left(\rmi
\frac{q_2}{q_1}\right)=\frac{q_2}{q_1^4},\\
  V_{-4}(q_1,q_2)&=q_1^{-4}P_2^{\left(\frac{1}{2},\frac{1}{2}\right)}\left(\rmi
      \frac{q_2}{q_1}\right)=\frac{q_1^2+4q_2^2}{q_1^6}.
  \end{split}
  \]
  with the corresponding first integrals
  \[
  \begin{split}
    I_{-3}(q_1,q_2,p_1,p_2)&=p_1(q_2p_1-q_1p_2)+\frac{q_1^2+4q_2^2}{2q_1^4},\\
  I_{-4}(q_1,q_2,p_1,p_2)&=p_1(q_2p_1-q_1p_2)+\frac{4q_2(q_1^2+2q_2^2)}{q_1^6}.
  \end{split}
  \]
\end{example}
In papers \cite{Ramani:82::,Dorizzi:83::} authors shown that
potential~\eqref{s1} is integrable for all $n\in\Z$. The first
integral of the Hamiltonian system with potential~\eqref{s1} is the
following
\[
I(q_1,q_2,p_1,p_2)=p_1(q_2p_1-q_1p_2)+\frac{1}{2}q_1^2V_{n-1}.
\]

Failures of  the direct search of additional first integrals for
other small values of $p$ and $\kappa$ suggests the non-integrability
for $p\neq 2$.  Since all known integrability conditions due to
variational equations are satisfied the only tool for proving the
non-integrability are higher order variational equations. It is known
result that if the identity component of differential Galois group of
$m$-th order variational equations for some $m\geq1$ is non-Abelian,
then Hamiltonian system is not integrable in the Liouville
sense~\cite{Morales:99::c}.  But the dimension of homogeneous linear
systems corresponding to higher order variational equations grows very
quickly with growing $m$ and determination of its differential Galois group
becomes a very hard problem. Thus in practice usually we look for
logarithms in solutions of first as well as higher order variational
equations. In some cases one can prove that the presence of logarithms
in solutions of higher order variational equations of a certain degree $m\geq 1$
implies that their differential Galois group has non-Abelian identity
component and as result to prove strictly the non-integrability. This
is the case when the first order variational equations are the direct
sum of Lam\'e equations \cite{Morales:99::c,mp:04::d,mp:05::c} or
when linear homogeneous equations are completely reducible
\cite{Boucher:00::}. None of these cases holds for Hamiltonian systems
governed for potentials of the form~\eqref{eq:dorki} but one can check
for small $p$ and $\kappa$ that for $p\neq 2$ logarithmic terms in
solutions of higher order variational equations appear, in contrary
for $p=2$, there is no such terms.

\subsection{Example with  $k\in\Z_+\setminus \{0,2\}$}
\label{k+}
In this section we consider homogeneous potentials of  a positive  degree. 
Similarly to the previous case we assume that potential
possesses maximal number $r+1$ of simple proper Darboux points
different from $\pm\rmi$, and all eigenvalues $\lambda_i$ for
$i=1,\ldots, r+1$ belong to first item of table~\eqref{tabMo}, i.e., 
\[
\lambda_i\in\left\lbrace p+\frac{1}{2}kp(p-1) \ | \ p\in\Z
\right\rbrace.
\]
so, we have
\begin{equation}
  \label{c21}
  \Lambda_i:=\lambda_i-1\in \mathcal{L}_k,\qquad 
 \mathcal{L}_k:=\left\lbrace \frac{1}{2}(p-1)(kp+2) \ | \ p\in\Z
\right\rbrace.
\end{equation}
Set $\mathcal{L}_k$ possesses only one negative element equal to $-1$, so
to satisfy relation~\eqref{eq:relpotek}
 we must
take at least two negative values, for example
$\Lambda_r=\Lambda_{r+1}=-1$ and $\Lambda_i>0$ for
$i=1,\ldots,r-1$. We assume also that
$\Lambda_1=\Lambda_2=\ldots=\Lambda_{r-1}=\Lambda=(p-1)(kp+2)/2$, then
obviously
\begin{equation}
  r-1=\frac{1}{2}(p-1)(kp+2).
  \label{eq:rr2}
\end{equation}
We can decompose quotient $h(z)/g(z)$ as follows
\[
\frac{h(z)}{g(z)}=\frac{1}{\Lambda}\sum_{i=1}^{r-1}\frac{1}{z-z_i}-\frac{1}{z-z_
  r}-\frac{1}{z-z_{r+1}}=\frac{1}{\Lambda}\sum_{i=1}^{r+1}\frac{1}{z-z_i}
-\left(1+\frac{1}{\Lambda}\right)\left(
  \frac{1}{z-z_r}+\frac{1}{z-z_{r+1}}\right),
\]
where affine coordinates of Darboux points $z_r$ and $z_{r+1}$ play role
of parameters. Now using~\eqref{c2} and~\eqref{c21} we obtain
\[
\begin{split}
  h(z)&=\frac{1}{\Lambda}g'(z)-\left(1+\frac{1}{\Lambda}\right)\left(\frac{1}{z-z_
      r}+\frac{1}{z-z_{r+1}}\right)g(z)\\
  &=\frac{2}{(p-1)(2+kp)}g'(z)-\frac{p(2+k(p-1))}{(p-1)(2+kp)}\left(\frac{1}{z-z_r
    } +\frac{1}{z-z_{r+1}}\right)g(z).
\end{split}
\]
  From definitions of $g(z)$ and $h(z)$ given 
in~\eqref{funGH} we obtain second order differential equation
\begin{equation}
  \label{c22}
  \begin{split}
    &v''(z)+p(z)v'(z)+q(z)v(z)=0,\\
    &p(z)=-\dfrac{p (2 - k + k p)}{2 (z-a )}-\dfrac{p (2 - k + k p)}{2
      (z-b)}
    + \dfrac{(1 + p) (2 - 2 k + k p) z}{ 2 (1 + z^2)},\\
    & q(z)=\dfrac{k (2 + k (-1 + p)) p (a b - z^2)}{2 (z-a ) (z-b) (1
      + z^2)},
  \end{split}
\end{equation}
where $a=z_r$ and $b=z_{r+1}$.  Now we have to find its solutions which 
are polynomials of degree $r$. This degree, for given  $p$ and $k$, is defined 
 by equation \eqref{eq:rr2}.     Substitution of the expansion
\[
v(z)=\sum_{i=0}^rv_iz^i,\qquad v_i\in\C.
\]
into \eqref{c22} gives a system of homogeneous linear equations of
dimension $r+1$ on indefinite coefficients $v_i$ for $i=0,\ldots,r$. This
system reads
\begin{equation}
  \vA\vv=\vzero,\qquad \vv=[v_0,\ldots,v_r]^T.
  \label{eq:coeffv2}
\end{equation}
Entries of matrix $\vA$ depend on $a$ and $b$. This system has a non-zero
solution iff 
all minors of $\vA$ of degree $r+1$ vanish. These conditions have the form of
non-linear  polynomial equations for $a$ and $b$. Among all solutions of these
equations we choose those which give  as solution $v(z)$ a polynomial  of 
degree $r$.  Then we  calculate
$V(q_1,q_2)=q_1^kv(q_2/q_1)$.  In order to obtain a rational
potential which is not a polynomial one, the number $s=r-k$  must be greater
than zero. This   gives the following condition on
$k$ and $p$
\[
s=p + \dfrac{k}{2} (p-2) (1 + p)>0.
\]
Thus, let us start our analysis for small $k$, and  small $|p|$. For $k=1$ and
$p\in\{-2,-1,0,1\}$ there is no any  non-trivial solution of~\eqref{eq:coeffv2}.
For $p=2$, and
$p=-3$, system \eqref{eq:coeffv2} possesses non-trivial solutions with
$r+1=4$ proper Darboux points. But all obtained potentials possess some multiple
proper Darboux points or  $a=\rmi$ or $b=\rmi$. For $p=3$ one
can find the following potential satisfying all assumptions
\begin{equation}
V=q_1^{-5}(147 q_1^6 + 441 q_1^4 q_2^2 + 56 q_1^2 q_2^4 + 4 q_2^6).
\label{potprut}
\end{equation}
Second order variational equations for this potential possess
logarithmic terms that suggests its non-integrability. But analysis of
this case is not finished because equations on unknowns $a$ and $b$ obtained
from condition of
vanishing of all minors of $\vA$ of degree $r+1$ are very complicated
polynomials of
$a$ and $b$  of degree 11.  Thus, the problem of finding all their
solutions is very hard. Potential \eqref{potprut} is just one example
corresponding to the simplest solution of this system.

 For $k=3$ and $p\in\{-1,0,1\}$ system
\eqref{eq:coeffv2} does not possesses  solution which  gives a potential with
required properties. For $p=2$ obtained potentials of required degree
$r=5$ have very complicated coefficients and verification if  they
possess all required properties as well as their effective subsequent
analysis seems to be impossible.

These  examples show that the reconstruction of the potential from the
given non-trivial eigenvalues of Hessians at all proper Darboux points
is very hard task. Furthermore, if we are lucky and we made this
reconstruction effectively, then we obtain potential that satisfies
all known necessary integrability conditions but very often it is
non-integrable.
\subsection{Potentials without proper Darboux points}
Another problem is related  to the integrability analysis of rational
potentials without proper Darboux points given in equation~\eqref{eq:bezp}. For
them we have no integrability obstructions.  Thus at this moment the only tool
that we have to our disposal is the direct method of search of first integrals
\cite{MR879243}. Such analysis for this class of potentials was already
done.   Namely, if we make canonical transformation
\[
z_1=q_2-\rmi q_1,\quad z_2=q_2+\rmi q_1, \qquad
y_1=\frac{1}{2}(p_2+\rmi p_1),\quad y_2=\frac{1}{2}(p_2-\rmi p_1),
\]
then Hamiltonian 
\[
H=\dfrac{1}{2}(p_1^2+p_2^2)+c(q_1+\rmi q_2)^{\alpha}(q_1-\rmi q_2)^{\beta},
\quad
    \alpha+\beta=k, \quad \alpha,\beta\in \Z, \qquad c\in\C^{\star}
\]
takes the form
\begin{equation}
  H=2y_1y_2+c z_1^{\alpha}z_2^{\beta}.
\end{equation} 
Integrability analysis for this class of Hamiltonian systems  with not only 
integer but also with rational $\alpha$ and $\beta$ was
already
made in \cite{mp:11::c} and many integrable rational potentials were
found.

\section*{Acknowledgments}
The authors are very grateful to Andrzej J. Maciejewski for many helpful
comments and suggestions concerning improvements and simplifications of some
results.

\newcommand{\noopsort}[1]{}\def\polhk#1{\setbox0=\hbox{#1}{\ooalign{\hidewidth
  \lower1.5ex\hbox{`}\hidewidth\crcr\unhbox0}}} \def\cprime{$'$}
  \def\cydot{\leavevmode\raise.4ex\hbox{.}} \def\cprime{$'$} \def\cprime{$'$}
  \def\polhk#1{\setbox0=\hbox{#1}{\ooalign{\hidewidth
  \lower1.5ex\hbox{`}\hidewidth\crcr\unhbox0}}} \def\cprime{$'$}
  \def\cprime{$'$} \def\cprime{$'$} \def\cprime{$'$} \def\cprime{$'$}
  \def\cprime{$'$}


\begin{thebibliography}{18}
\providecommand{\natexlab}[1]{#1}
\providecommand{\url}[1]{\texttt{#1}}
\expandafter\ifx\csname urlstyle\endcsname\relax
  \providecommand{\doi}[1]{doi: #1}\else
  \providecommand{\doi}{doi: \begingroup \urlstyle{rm}\Url}\fi

\bibitem[Boucher(2000)]{Boucher:00::}
D.~Boucher.
\newblock \emph{Sur des \'equations diff\'erentielles lin\'eares
  param\'etr\'ees, une application aux syst\'emes hamiltoniens}.
\newblock PhD thesis, Universit\'e de Limoges, France, 2000.

\bibitem[Casale et~al.(2010)Casale, Duval, Maciejewski, and
  Przybylska]{mp:10::a}
G.~Casale, G. Duval, A.~J. Maciejewski, and M. Przybylska.
\newblock Integrability of {H}amiltonian systems with homogeneous potentials of
  degree zero.
\newblock \emph{Phys. Lett. A}, 374\penalty0 (3):\penalty0 448--452, 2010.

\bibitem[Dorizzi et~al.(1983)Dorizzi, Grammaticos, and Ramani]{Dorizzi:83::}
B.~Dorizzi, B.~Grammaticos, and A.~Ramani.
\newblock A new class of integrable systems.
\newblock \emph{J. Math. Phys.}, 24\penalty0 (9):\penalty0 2282--2288, 1983.


\bibitem[Duval and Maciejewski(2009)]{Maciejewski:09::}
G.~Duval and A.~J. Maciejewski.
\newblock Jordan obstruction to the integrability of homogeneous potentials.
\newblock \emph{Ann. Inst. Fourier}, 59\penalty0 (7):\penalty0 2839--2890,
  2009.

\bibitem[Hietarinta(1987)]{MR879243}
J. Hietarinta.
\newblock Direct methods for the search of the second invariant.
\newblock \emph{Phys. Rep.}, 147\penalty0 (2):\penalty0 87--154, 1987.


\bibitem[Kimura(1969/1970)]{Kimura:69::}
T. Kimura.
\newblock On {R}iemann's equations which are solvable by quadratures.
\newblock \emph{Funkcial. Ekvac.}, 12:\penalty0 269--281, 1969/1970.

\bibitem[Kovacic(1986)]{Kovacic:86::}
J.~J. Kovacic.
\newblock An algorithm for solving second order linear homogeneous differential
  equations.
\newblock \emph{J. Symbolic Comput.}, 2\penalty0 (1):\penalty0 3--43, 1986.

\bibitem[Maciejewski and Przybylska(2004)]{mp:04::d}
A.~J. Maciejewski and M. Przybylska.
\newblock All meromorphically integrable 2{D} {H}amiltonian systems with
  homogeneous potentials of degree 3.
\newblock \emph{Phys. Lett. A}, 327\penalty0 (5-6):\penalty0 461--473, 2004.

\bibitem[Maciejewski and Przybylska(2005)]{mp:05::c}
A.~J. Maciejewski and M. Przybylska.
\newblock Darboux points and integrability of {H}amiltonian systems with
  homogeneous polynomial potential.
\newblock \emph{J. Math. Phys.}, 46\penalty0 (6):\penalty0 062901, 33 pages,
  2005.

\bibitem[Maciejewski et~al.(2011)Maciejewski, Przybylska, and
  Tsiganov]{mp:11::c}
A.~J. Maciejewski, M. Przybylska, and A.V. Tsiganov.
\newblock On algebraic construction of certain integrable and super-integrable
  systems.
\newblock \emph{Physica D}, 240\penalty0 (18):\penalty0 1426--1448, 2011.

\bibitem[Morales~Ruiz(1999)]{Morales:99::c}
J.~J. Morales~Ruiz.
\newblock \emph{Differential {G}alois theory and non-integrability of
  {H}amiltonian systems}, volume 179 of \emph{Progress in Mathematics}.
\newblock Birkh\"auser Verlag, Basel, 1999.


\bibitem[Morales-Ruiz and Ramis(2001)]{Morales:01::a}
J.~J. Morales-Ruiz and J.-P. Ramis.
\newblock A note on the non-integrability of some {H}amiltonian systems with a
  homogeneous potential.
\newblock \emph{Methods Appl. Anal.}, 8\penalty0 (1):\penalty0 113--120, 2001.


\bibitem[Przybylska(2009)]{mp:09::a}
M. Przybylska.
\newblock Darboux points and integrability of homogenous {H}amiltonian systems
  with three and more degrees of freedom.
\newblock \emph{Regul. Chaotic Dyn.}, 14\penalty0 (2):\penalty0 263--311, 2009.

\bibitem[Ramani et~al.(1982)Ramani, Dorizzi, and Grammaticos]{Ramani:82::}
A.~Ramani, B.~Dorizzi, and B.~Grammaticos.
\newblock Painlev\'e conjecture revisited.
\newblock \emph{Phys. Rev. Lett.}, 49\penalty0 (21):\penalty0 1539--1541, 1982.


\bibitem[Inc.()]{Wolfram}
Wolfram~Research Inc.
\newblock The {W}olfram {F}unctions {S}ite: functions.wolfram.com.
\newblock URL \url{http://functions.wolfram.com/Polynomials/JacobiP/}.

\bibitem[Studzi\'nski(2009)]{Studzinski:09::}
M.~Studzi\'nski.
\newblock \emph{Integrability studies of a certain class of homogeneous
  rational potentials}.
\newblock B.Sc. thesis, N.~Copernicus University, Toru\'n, Poland, 2009.
\newblock in Polish.

\bibitem[Studzi\'nski(2010)]{Studzinski:10::}
M.~Studzi\'nski.
\newblock \emph{Integrability studies of homogeneous rational potentials}.
\newblock M.Sc. thesis, N.~Copernicus University, Toru\'n, Poland, 2010.
\newblock in Polish.

\bibitem[Yoshida(1987)]{Yoshida:87::a}
H. Yoshida.
\newblock A criterion for the nonexistence of an additional integral in
  {H}amiltonian systems with a homogeneous potential.
\newblock \emph{Phys. D}, 29\penalty0 (1-2):\penalty0 128--142, 1987.


\end{thebibliography}
\end{document}